\documentclass[12pt]{article}

\usepackage{mystyle}
\usepackage[paper=letterpaper,margin=1in]{geometry}
\usepackage{lmodern}
\usepackage{enumerate}
\usepackage[inline,shortlabels]{enumitem}
\usepackage[ampersand]{easylist}
\usepackage{amsthm}
\usepackage{times}
\usepackage[comma, numbers, sort&compress, sectionbib]{natbib}
\usepackage[displaymath,floats,textmath,graphics,sections]{preview}

\newcommand{\todo}[1]{\textbf{TO DO:} #1 \textbf{End of TO DO.}}
\PreviewMacro*\todo
\newcommand{\comment}[1]{}
\PreviewMacro*\comment
\newcommand{\Formatting}[1]{#1}

\newcommand{\Id}[1]{\mathit{#1}}

\newcommand{\cbr}[1]{\left\{ #1 \right\}}
\newcommand{\abr}[1]{\langle #1 \rangle}

\newcommand{\Set}[1]{\cbr{#1}}
\newcommand{\SetBuild}[2]{\Set{#1 \mid #2}}
\newcommand{\Seq}[1]{\abr{#1}}

\newcommand{\Car}[1]{\left\vert #1 \right\vert}

\newcommand{\Oh}[1]{{O}({#1})}

\theoremstyle{definition} \newtheorem{lemma}{Lemma}
\theoremstyle{definition} 
\theoremstyle{definition} \newtheorem{theorem}{Theorem}

\title{Strategyproof Pareto-Stable Mechanisms for Two-Sided Matching with Indifferences\thanks{Department of Computer Science,
  University of Texas at Austin,
  2317 Speedway, Stop D9500,
  Austin, Texas 78712--1757.
  Email: \{onur,\,geocklam,\,plaxton\}@cs.utexas.edu.
  This research was supported by NSF Grant CCF--1217980.
  An earlier version of this work appears as a UTCS technical report~\cite{domanic+lp:match}.}}

\author{Nevzat Onur Domani\c{c}
  \and Chi-Kit Lam
  \and C. Gregory Plaxton
}

\date{March 2017}

\begin{document}

\begin{titlepage}
\maketitle
\thispagestyle{empty}

%%% Local Variables:
%%% mode: latex
%%% TeX-master: "tr"
%%% End:

\begin{abstract}
  We study variants of the stable marriage and college admissions
  models in which the agents are allowed to express weak preferences
  over the set of agents on the other side of the market and the
  option of remaining unmatched.  For the problems that we address,
  previous authors have presented polynomial-time algorithms for
  computing a ``Pareto-stable'' matching.  In the case of college
  admissions, these algorithms require the preferences of the colleges
  over groups of students to satisfy a technical condition related to
  responsiveness.  We design new polynomial-time Pareto-stable
  algorithms for stable marriage and college admissions that correspond to strategyproof mechanisms.  For stable
  marriage, it is known that no Pareto-stable mechanism is
  strategyproof for all of the agents; our algorithm provides a
  mechanism that is strategyproof for the agents on one side of the
  market.  For college admissions, it is known that no Pareto-stable
  mechanism can be strategyproof for the colleges; our algorithm
  provides a mechanism that is strategyproof for the students.
\end{abstract}

%%% Local Variables:
%%% mode: latex
%%% TeX-master: "tr"
%%% End:

\end{titlepage}

%!TEX root = tr.tex

\section{Introduction}
\label{sec:intro}

Gale and Shapley~\cite{gale+s:match} introduced the stable marriage
model and its generalization to the college admissions model.  Their
work spawned a vast literature on two-sided matching; see
Manlove~\cite{manlove:match} for a recent survey.  The present paper
is primarily concerned with variants of the stable marriage and
college admissions models where the agents have weak preferences,
i.e., where indifferences are allowed.

In the most basic stable marriage model, we are given an equal number
of men and women, where each man (resp., woman) has complete, strict
preferences over the set of women (resp., men); we refer to this model
as SMCS.  For SMCS, an outcome is a matching that pairs up all of the
men and women into disjoint man-woman pairs. A man-woman pair $(p,q)$
is said to form a \emph{blocking pair} for a matching $M$ if $p$
prefers $q$ to his partner in $M$ and $q$ prefers $p$ to her partner
in $M$.  A matching is \emph{stable} if it does not have a blocking
pair.  It is straightforward to prove that any stable matching is also
Pareto-optimal.  Gale and Shapley presented the deferred acceptance
(DA) algorithm for the SMCS problem and proved that the man-proposing
version of the DA algorithm produces the unique man-optimal (and
woman-pessimal) stable matching.  Roth~\cite{roth:match82} showed that
the associated mechanism, which we refer to as the \emph{man-proposing
  DA mechanism}, is \emph{strategyproof} for the men, i.e., it is a
weakly dominant strategy for each man to declare his true preferences.
Unfortunately, the man-proposing DA mechanism is not strategyproof for
the women.  In fact, Roth~\cite{roth:match82} showed that no stable
mechanism for SMCS is strategyproof for all of the agents.

The SMCW model is the generalization of the SMCS model in which each
man (resp., woman) has weak preferences over the set of women (resp.,
men).  When indifferences are allowed, we need to refine our notion of
a blocking pair.  A man-woman pair $(p,q)$ is said to form a
\emph{strongly blocking pair} for a matching $M$ if $p$ prefers $q$ to
his partner in $M$ and $q$ prefers $p$ to her partner in $M$.  A
matching is \emph{weakly stable} if it is individually rational
and does not have a strongly
blocking pair.  Two other natural notions of stability, namely strong
stability and super-stability, have been investigated in the
literature (see Manlove~\cite[Chapter~3]{manlove:match} for a survey
of these results).  We focus on weak stability because every SMCW
instance admits a weakly stable matching (this follows from the existence of
stable matchings for SMCS, coupled with arbitrary tie-breaking), but not every
SMCW instance admits a strongly stable or super-stable matching.  It
is straightforward to exhibit SMCW instances (with as few as two men
and two women) for which some weakly stable matching is not Pareto-optimal.
Sotomayor~\cite{sotomayor:pareto-stable} argues that \emph{Pareto-stability} (i.e., Pareto-optimality plus weak stability) is an appropriate solution concept for SMCW and certain other matching models with weak preferences, and proves that every SMCW instance admits a Pareto-stable matching.

Erdil and Ergin~\cite{erdil+e:match} and Chen and
Ghosh~\cite{chen+g:match} present polynomial-time algorithms for
computing a Pareto-stable matching of a given SMCW instance; in fact,
these algorithms are applicable to certain more general models to be
discussed shortly.  Given the existence of a stable mechanism for SMCS
that is strategyproof for the men (or, symmetrically, for the women),
it is natural to ask whether there is a Pareto-stable mechanism for
SMCW that is strategyproof for the men.  We cannot hope to find a
Pareto-stable mechanism for SMCW that is strategyproof for all agents,
since that would imply a stable mechanism for SMCS that is
strategyproof for all agents.  A similar statement holds for the SMIW
model, the generalization of the SMCW model in which the agents are
allowed to express incomplete preferences.  See Section~\ref{sec:smiw}
for a formal definition of the SMIW model and the associated notions
of weak stability and Pareto-stability.  Throughout the remainder of
the paper, when we say that a mechanism for a stable marriage model is
strategyproof, we mean that it is strategyproof for the agents on one
side of the market; moreover, unless otherwise specified, it is to be
understood that the mechanism is strategyproof for the men.
The Pareto-stable algorithms of Erdil and Ergin, and of Chen and
Ghosh, are based on a two-phase approach where the first phase runs
the Gale-Shapley DA algorithm after breaking all ties arbitrarily.  In
Appendix~\ref{app:phase} we show that this approach does not provide a
strategyproof mechanism.

This paper provides the first Pareto-stable mechanism for SMIW (and
also SMCW) that is shown to be strategyproof.  We present a
nondeterministic algorithm for SMIW that generalizes Gale and
Shapley's DA algorithm as follows: in each iteration, a
nondeterministically chosen unmatched man ``proposes'' simultaneously
to all of the women in his next-highest tier of preference (i.e., the
highest tier to which he has not already proposed); the women respond
to this proposal by solving a certain maximum-weight matching problem
to determine which man becomes unmatched (i.e., the man making the
proposal or one of the tentatively matched men).  Our generalization
of the DA mechanism admits a polynomial-time implementation.

The college admissions model with weak preferences, which we denote
CAW, is a further generalization of the SMIW model.  In the CAW model,
students and colleges are being matched rather than men and women, and
each college has a positive integer capacity representing the number
of students that it can accommodate.  See Section~\ref{sec:caw} for a
formal definition of the CAW model and the associated notions of weak
stability and Pareto-stability.

A key difference between CAW and SMIW is that in addition to
expressing preferences over individual students, the colleges have
preferences over \emph{groups} of students.  This characteristic is
shared by the CAS model, which is the restriction of the CAW model to
strict preferences.  It is known that no stable mechanism for CAS is
strategyproof for the colleges~\cite{roth:match85}; the proof makes
use of the fact that the colleges do not (in general) have unit
demand.  It follows that no Pareto-stable mechanism for CAW is
strategyproof for the colleges.  Throughout the remainder of the
paper, when we say that a mechanism for a college admissions model is
strategyproof, we mean that it is strategyproof for the students.

\comment{
  TODO: CAN WE CITE A RESULT ADDRESSING THE NECESSITY OF RESPONSIVENESS?
}

Gale and Shapley's DA algorithm generalizes easily to the CAS model.
Roth~\cite{roth:match85} has shown that the student-proposing DA
algorithm provides a strategyproof stable mechanism for CAS when the
preferences of the colleges are \emph{responsive}.  When the
colleges have responsive preferences, the student-proposing DA
mechanism is also known to be student-optimal for
CAS~\cite{roth:match85}.

Erdil and Ergin~\cite{erdil+e:match} consider the special case of the
CAW model where the following restrictions hold for all students $x$
and colleges $y$: $x$ is not indifferent between being assigned to $y$
and being left unassigned; $y$ is not indifferent between having one
of its slots assigned to $x$ and having that slot left unfilled.  We
remark that this special case of CAW corresponds to the HRT problem
discussed in Manlove~\cite[Chapter~3]{manlove:match}.\footnote{In the
  model of Erdil and Ergin, which is stated using worker-firm
  terminology rather than student-college terminology, a ``no
  indifference to unemployment/vacancy'' assumption makes the
  aforementioned restrictions explicit.  In the HRT model of Manlove,
  which is stated using resident-hospital terminology rather than
  student-college terminology, it is assumed that a set of acceptable
  resident-hospital pairs is given, and that each agent specifies weak
  preferences over the set of agents with whom they form an acceptable
  pair.  We consider the approach of Erdil-Ergin --- where the
  starting point is the preferences of the individual agents, and the
  ``acceptability'' of a given pair of agents may be deduced from
  those preferences --- to be more natural, but the resulting models
  are equivalent.}  For this special case, Erdil and Ergin present a
polynomial-time algorithm for computing a Pareto-stable matching when
the preferences of the colleges satisfy a technical restriction
related to responsiveness.  We consider the same class of preferences,
which we refer to as \emph{minimally responsive}; see
Section~\ref{sec:caw} for a formal definition.  The algorithm of Erdil
and Ergin does not provide a strategyproof mechanism.  Chen and
Ghosh~\cite{chen+g:match} build on the results of Erdil and Ergin by
considering the many-to-many generalization of HRT in which the agents
on both sides of the market have capacities (and the agent preferences are minimally responsive).  For this generalization, Chen
and Ghosh provide a \emph{strongly} polynomial-time algorithm.  No
strategyproof mechanism (even for the agents on one side of the
market) is possible in the many-to-many setting, since it is a
generalization of CAS.
% Moreover, if the
% algorithm of Chen and Ghosh is specialized to the case where the
% agents on one side of the market have unit capacity, it does not
% provide a strategyproof mechanism.
We provide the first Pareto-stable mechanism for CAW that is shown to
be strategyproof.  As in the work of Erdil-Ergin and Chen-Ghosh, we
assume that the preferences of the colleges are minimally
responsive.  We can also handle the class of college preferences
``induced by additive utility'' that is defined in
Section~\ref{sec:caw-further-disc}.

In the many-to-many matching setting addressed by Chen and
Ghosh~\cite{chen+g:match}, a pair of agents (on opposite sides of the
market) can be matched with arbitrary multiplicity, as long as the
capacity constraints are respected.  Chen~\cite{chen:match12} presents
a polynomial-time algorithm for the variation of many-to-many matching
in which a pair of agents can only be matched with multiplicity one.
Kamiyama~\cite{kamiyama:match14} addresses the same problem using a
different algorithmic approach.  (The algorithms of Chen and Kamiyama
are strongly polynomial, since we can assume without loss of
generality that the capacity of any agent is at most the number of
agents on the other side of the market.)  Since this variation of the
many-to-many setting also generalizes CAS, it does not admit a
strategyproof mechanism, even for the agents on one side of the
market.

Erdil and Ergin~\cite{erdil+e:tie,erdil+e:match}
and Kesten~\cite{Kes10} consider a second
natural solution concept in addition to Pareto-stability.  In the
context of SMIW (or its special case SMCW), this second solution
concept seeks a weakly stable matching $M$ that is ``man optimal'' in
the following sense: for all weakly stable matchings $M'$, either all
of the men are indifferent between $M$ and $M'$, or at least one man
prefers $M$ to $M'$.  Erdil and Ergin~\cite{erdil+e:match} present a
polynomial-time algorithm to compute such a man optimal weakly stable
matching for SMIW; in fact, their algorithm is presented for the
generalization of SMIW to CAW.  Erdil and Ergin~\cite{erdil+e:tie}
and Kesten~\cite{Kes10}
prove that no strategyproof man optimal weakly stable mechanism exists
for SMCW.
% for the special case of SMCW in which the men have strict
% preferences.  Using a similar argument, we prove in
% Appendix~\ref{app:???} that no strategyproof man optimal weakly
% stable mechanism exists for the special case of SMCW in which the
% women have strict preferences.
Prior to our work, it was unclear whether such an impossibility result
might hold for strategyproof Pareto-stable mechanisms for SMCW (or its
generalizations to SMIW and CAW).

The assignment game of Shapley and Shubik~\cite{shapley+s:assign} can
be viewed as an auction with multiple distinct items where each bidder
is seeking to acquire at most one item.  This class of
\emph{unit-demand auctions} has been heavily studied in the literature
(see, e.g., Roth and Sotomayor~\cite[Chapter~8]{roth+s:match}).  In
Section~\ref{sec:uap}, we define the notion of a ``unit-demand auction
with priorities'' (UAP) and establish a number of useful properties of
UAPs; these are straightforward generalizations of corresponding
properties of unit-demand auctions.  Section~\ref{sec:iuap} builds on
the UAP notion to define the notion of an ``iterated UAP'' (IUAP), and
establishes a number of important properties of IUAPs; these results
are nontrivial to prove and provide the technical foundation for our
main results.  Section~\ref{sec:smiw} presents our first main result,
a polynomial-time algorithm for SMIW that provides a strategyproof
Pareto-stable mechanism.  Section~\ref{sec:caw} presents our second
main result, a polynomial-time algorithm for CAW that provides a
strategyproof Pareto-stable mechanism assuming that the preferences of
the colleges are minimally responsive.

\section{Unit-Demand Auctions with Priorities}
\label{sec:uap}

In this section, we formally define the notion of a unit-demand auction with priorities (UAP).
In Section~\ref{sec:uap-matroid}, we describe an associated matroid for a given UAP and we use this matroid to define the notion of a ``greedy MWM''.
In Section~\ref{sec:uap-extend}, we establish a result related to extending a given UAP by introducing additional bidders.
In Section~\ref{sec:inc-hungarian-step}, we discuss how to efficiently compute a greedy MWM in a UAP.
In Section~\ref{sec:uap-threshold}, we introduce a key definition that is helpful for establishing our strategyproofness results.
We start with some useful definitions.

% Items
\newcommand{\Item}{v}
% Shortcuts for items
% with primes
\newcommand{\Itemp}{\Item'}
% with double primes
\newcommand{\Itempp}{\Item''}
% with subscripts
\newcommand{\ItemSub}[1]{\Item_{#1}}

\newcommand{\Reals}{\mathbb{R}}

% Item sets
\newcommand{\Items}{V}
% Shortcuts for item sets
% with primes
\newcommand{\Itemsp}{\Items'}

% Unit-demand bid
\newcommand{\UBid}{\beta}

% Offer (real number) in a unit-demand bid
\newcommand{\Offer}{x}
% Shortcuts
% with primes
\newcommand{\Offerp}{\Offer'}
% with subscripts
\newcommand{\OfferSub}[1]{\Offer_{#1}}
% optimal
\newcommand{\OfferOpt}{\Opt{\Offer}}

A \emph{(unit-demand) bid $\UBid$ for a set of items $\Items$} is a subset of $\Items \times \Reals$ such that no two pairs in $\UBid$ share the same first component.
(So $\UBid$ may be viewed as a partial function from $\Items$ to $\Reals$.)

% For symbols denoting (somewhat) optimal instances
\newcommand{\Opt}[1]{{#1}^*}

% Bidders
\newcommand{\Bidder}{u}
% Shortcuts for bidders
% with primes
\newcommand{\Bidderp}{\Bidder'}
% with double primes
\newcommand{\Bidderpp}{\Bidder''}
% with subscripts
\newcommand{\BidderSub}[1]{\Bidder_{#1}}
% optimal
\newcommand{\BidderOpt}{\Opt{\Bidder}}

% Components of bidders:
\newcommand{\BidId}{\alpha} % Id
\newcommand{\Pri}{z} % Priority
% Shortcuts
% with primes
\newcommand{\Prip}{\Pri'}
% with double primes
\newcommand{\Pripp}{\Pri''}
% with subscripts
\newcommand{\BidIdSub}[1]{\BidId_{#1}}
\newcommand{\PriSub}[1]{\Pri_{#1}}
% optimal
\newcommand{\PriOpt}{\Opt{\Pri}}

% Functions for components of bidders
\newcommand{\BId}[1]{\Id{id}(#1)}
\newcommand{\BUBid}[1]{\Id{bid}(#1)}
\newcommand{\BPriId}{\Id{priority}}
\newcommand{\BPri}[1]{\BPriId(#1)}

% Items that a bidder bids for
\newcommand{\BItemsId}{\Id{items}}
\newcommand{\BItems}[1]{\BItemsId(#1)}

A \emph{bidder $\Bidder$ for a set of items $\Items$} is a triple $(\BidId, \UBid, \Pri)$ where $\BidId$ is an integer ID, $\UBid$ is a bid for $\Items$, and $\Pri$ is a real priority.
For any bidder $\Bidder = (\BidId, \UBid, \Pri)$, we define $\BId{\Bidder}$ as $\BidId$, $\BUBid{\Bidder}$ as $\UBid$, $\BPri{\Bidder}$ as $\Pri$, and $\BItems{\Bidder}$ as the union, over all $(\Item, \Offer)$ in $\UBid$, of $\Set{\Item}$.

% Bidder sets
\newcommand{\Bidders}{U}
% Shortcuts for bidder sets
% with primes
\newcommand{\Biddersp}{\Bidders'}
% with double primes
\newcommand{\Bidderspp}{\Bidders''}
% with subscripts
\newcommand{\BiddersSub}[1]{\Bidders_{#1}}

% UAPs
\newcommand{\Auc}{A}
% Shortcuts for UAPs
% with primes
\newcommand{\Aucp}{\Auc'}
% with double primes
\newcommand{\Aucpp}{\Auc''}
% with subscripts
\newcommand{\AucSub}[1]{\Auc_{#1}}

A \emph{unit-demand auction with priorities (UAP)} is a pair $\Auc = (\Bidders, \Items)$ satisfying the following conditions:
$\Items$ is a set of items;
$\Bidders$ is a set of bidders for $\Items$;
each bidder in $\Bidders$ has a distinct ID\@.

\subsection{An Associated Matroid}
\label{sec:uap-matroid}

% Edges and edge sets
\newcommand{\Edge}{e}
\newcommand{\Edges}{E}

% Weight of an edge, edge set, bidder item pair, matching, and a path
\newcommand{\WEdge}[1]{\Id{w}(#1)}
\newcommand{\WEdges}[1]{\Id{w}(#1)}
\newcommand{\WPair}[2]{\Id{w}(#1, #2)}
\newcommand{\WMatch}[1]{\WEdges{#1}}
\newcommand{\WPath}[1]{\Id{w}(#1)}

% Weight of an MWM of an UAP
\newcommand{\WAuc}[1]{\Id{w}(#1)}

A UAP $\Auc = (\Bidders, \Items)$ may be viewed as an edge-weighted bipartite graph, where the set of edges incident on bidder $\Bidder$ correspond to $\BUBid{\Bidder}$:
for each pair $(\Item, \Offer)$ in $\BUBid{\Bidder}$, there is an edge $(\Bidder, \Item)$ of weight $\Offer$.
We refer to a matching (resp., maximum-weight matching (MWM), maximum-cardinality MWM (MCMWM)) in the associated edge-weighted bipartite graph as a matching (resp., MWM, MCMWM) of $\Auc$.
For any edge $\Edge = (\Bidder, \Item)$ in a given UAP, the associated weight is denoted $\WEdge{\Edge}$ or $\WPair{\Bidder}{\Item}$.
For any set of edges $\Edges$, we define $\WEdges{\Edges}$ as $\sum_{\Edge \in \Edges} \WEdge{\Edge}$.
For any UAP $\Auc$, we let $\WAuc{\Auc}$ denote the weight of an MWM of $\Auc$.

% Matchings
\newcommand{\Match}{M}
% Shortcuts for matchings
% with primes
\newcommand{\Matchp}{\Match'}
% with double primes
\newcommand{\Matchpp}{\Match''}
% with triple primes
\newcommand{\Matchppp}{\Match'''}
% with subscripts
\newcommand{\MatchSub}[1]{\Match_{#1}}
% with subscripts and primes
\newcommand{\MatchSubp}[1]{\Match_{#1}'}
% optimal
\newcommand{\MatchOpt}{\Opt{\Match}}

% Matching Sets
\newcommand{\Matchings}{\mathcal{\Match}}
\newcommand{\Matchingsp}{\Matchings'}
\newcommand{\Matchingspp}{\Matchings''}
\newcommand{\Matchingsppp}{\Matchings'''}

\newcommand{\Indeps}{\mathcal{I}}

\begin{lemma}
  \label{lem:matroid}
  Let $\Auc = (\Bidders, \Items)$ be a UAP, and let $\Indeps$ denote the set of all subsets $\Biddersp$ of $\Bidders$ such that there exists an MWM of $\Auc$ that matches every bidder in $\Biddersp$.
  Then $(\Bidders, \Indeps)$ is a matroid.
\end{lemma}

% Vertex
\newcommand{\Vertex}{y}

% Paths, cycles, path sets, etc
\newcommand{\Path}{P}
\newcommand{\PathOther}{Q}
\newcommand{\Cycle}{C}
\newcommand{\PathSet}{\mathcal{\Path}}
\newcommand{\PathCycleColl}{\mathcal{S}}

\begin{proof}
  The only nontrivial property to show is the exchange property.
  Assume that $\BiddersSub{1}$ and $\BiddersSub{2}$ belong to $\Indeps$ and that $\Car{\BiddersSub{1}} > \Car{\BiddersSub{2}}$.
  Let $\MatchSub{1}$ be an MWM of $\Auc$ that matches every bidder in $\BiddersSub{1}$, and let $\MatchSub{2}$ be an MWM of $\Auc$ that matches every bidder in $\BiddersSub{2}$.
  If $\MatchSub{2}$ matches some bidder $\Bidder$ in $\BiddersSub{1} \setminus \BiddersSub{2}$, then $\BiddersSub{2} + \Bidder$ belongs to $\Indeps$, as required.
  Thus, in what follows, we assume that $\MatchSub{2}$ does not match any bidder in $\BiddersSub{1} \setminus \BiddersSub{2}$.
  The symmetric difference of $\MatchSub{1}$ and $\MatchSub{2}$, denoted $\MatchSub{1} \oplus \MatchSub{2}$, corresponds to a collection of vertex-disjoint paths and cycles.
  Since $\MatchSub{2}$ does not match any bidder in $\BiddersSub{1} \setminus \BiddersSub{2}$, we deduce that each bidder in $\BiddersSub{1} \setminus \BiddersSub{2}$ is an endpoint of one of the paths in this collection.
  Since $\Car{\BiddersSub{1}} > \Car{\BiddersSub{2}}$, $\Car{\BiddersSub{1} \setminus \BiddersSub{2}} = \Car{\BiddersSub{1}} - \Car{\BiddersSub{1} \cap \BiddersSub{2}}$, and $\Car{\BiddersSub{2} \setminus \BiddersSub{1}} = \Car{\BiddersSub{2}} - \Car{\BiddersSub{1} \cap \BiddersSub{2}}$, we have $\Car{\BiddersSub{1} \setminus \BiddersSub{2}} > \Car{\BiddersSub{2} \setminus \BiddersSub{1}}$.
  It follows that there is at least one path in this collection, call it $\Path$, such that one endpoint of $\Path$ is a bidder $\Bidder$ in $\BiddersSub{1}\setminus \BiddersSub{2}$ and the other endpoint of $\Path$ is a vertex $\Vertex$ that does not belong to $\BiddersSub{2} \setminus \BiddersSub{1}$.
  Moreover, $\Vertex$ does not belong to $\BiddersSub{1}$:
  if the length of $\Path$ is odd, then $\Vertex$ is an item and hence does not belong to $\BiddersSub{1}$;
  if the length of $\Path$ is even, then $\Vertex$ is not matched in $\MatchSub{1}$ and hence does not belong to $\BiddersSub{1}$.
  Since $\Vertex$ does not belong to $\BiddersSub{2} \setminus \BiddersSub{1}$ and does not belong to $\BiddersSub{1}$, we conclude that $\Vertex$ does not belong to $\BiddersSub{2}$.
  The edges of $\Path$ alternate between $\MatchSub{1}$ and $\MatchSub{2}$.
  Let $X_1$ denote the edges of $\Path$ that belong to $\MatchSub{1}$, and let $X_2$ denote the edges of $\Path$ that belong to $\MatchSub{2}$.
  Since $\MatchSub{1}$ is an MWM of $\Auc$ and $\MatchSubp{1} = \MatchSub{1} \oplus \Path = (\MatchSub{1} \cup X_2) \setminus X_1$ is a matching of $\Auc$, we deduce that $\WEdges{X_1} \geq \WEdges{X_2}$.
  Since $\MatchSub{2}$ is an MWM of $\Auc$ and $\MatchSubp{2} = \MatchSub{2} \oplus \Path = (\MatchSub{2} \cup X_1) \setminus X_2$ is a matching of $\Auc$, we deduce that $\WEdges{X_2} \geq \WEdges{X_1}$.
  Hence $\WEdges{X_1} = \WEdges{X_2}$ and $\MatchSubp{1}$ and $\MatchSubp{2}$ are MWMs of $\Auc$.
  The MWM $\MatchSubp{2}$ matches all of the vertices on $\Path$ except for $\Vertex$.
  Since $\Vertex$ does not belong to $\BiddersSub{2}$, we conclude that $\MatchSubp{2}$ matches all of the vertices in $\BiddersSub{2} + \Bidder$, and so the exchange property holds.
\end{proof}

\newcommand{\Matroid}[1]{\Id{matroid(#1)}}

For any UAP $\Auc$, we define $\Matroid{\Auc}$ as the matroid of Lemma~\ref{lem:matroid}.

For any UAP $\Auc = (\Bidders, \Items)$ and any independent set $\Biddersp$ of $\Matroid{\Auc}$, we define the \emph{priority of $\Biddersp$} as the sum, over all bidders $\Bidder$ in $\Biddersp$, of $\BPri{\Bidder}$.
For any UAP $\Auc$, the matroid greedy algorithm can be used to compute a maximum-priority maximal independent set of $\Matroid{\Auc}$.

\newcommand{\MatchedBidders}[1]{\Id{matched}(#1)}

\newcommand{\Unique}[1]{\Id{unique}(#1)}

For any matching $\Match$ of a UAP $\Auc = (\Bidders, \Items)$, we define $\MatchedBidders{\Match}$ as the set of all bidders in $\Bidders$ that are matched in $\Match$.
We say that an MWM $\Match$ of a UAP $\Auc$ is \emph{greedy} if $\MatchedBidders{\Match}$ is a maximum-priority maximal independent set of $\Matroid{\Auc}$.
For any UAP $\Auc$, we define the predicate $\Unique{\Auc}$ to hold if $\MatchedBidders{\Match} = \MatchedBidders{\Matchp}$ for all greedy MWMs $\Match$ and $\Matchp$ of $\Auc$.

\comment{
  Onur: We now use $\MatchedBidders{\Match}$ in the definition below.
}

\newcommand{\MPriority}[1]{\Id{priority}(#1)}

For any matching $\Match$ of a UAP, we define the \emph{priority of $\Match$}, denoted $\MPriority{\Match}$, as the sum, over all bidders $\Bidder$ in $\MatchedBidders{\Match}$, of $\BPri{\Bidder}$.
Thus an MWM is greedy if and only if it is a maximum-priority MCMWM\@.

\comment{
  Is it better to have: For each greedy MWM M and M', greedy(M, u) = greedy(M', u) ...
}

\begin{lemma}
  \label{lem:uap-distribution}
  All greedy MWMs of a given UAP have the same distribution of priorities.
\end{lemma}

\begin{proof}
  This is a standard matroid result that follows easily from the exchange property and the correctness of the matroid greedy algorithm.
\end{proof}

\newcommand{\Greedy}[2]{\Id{greedy}(#1, #2)}

For any UAP $\Auc$ and any real priority $\Pri$, we define $\Greedy{\Auc}{\Pri}$ as the (uniquely defined, by Lemma~\ref{lem:uap-distribution}) number of matched bidders with priority $\Pri$ in any greedy MWM of $\Auc$.

\begin{lemma}
  \label{lem:uap-in-out}
  Let $\Auc = (\Bidders, \Items)$ be a UAP\@.
  Let $\Bidder$ be a bidder in $\Bidders$ such that $(\Item, \Offer)$ belongs to $\BUBid{\Bidder}$, $\BPri{\Bidder} = \Pri$, and $\Bidder$ is not matched in any greedy MWM of $\Auc$.
  Let $\Bidderp$ be a bidder in $\Bidders$ such that $(\Item, \Offerp)$ belongs to $\BUBid{\Bidderp}$, $\BPri{\Bidderp} = \Prip$, and $\Bidderp$ is matched to $\Item$ in some greedy MWM of $\Auc$.
  Then $(\Offer, \Pri) < (\Offerp, \Prip)$.\footnote{Throughout this paper, comparisons of pairs are to be performed lexicographically.}
\end{lemma}

\begin{proof}
  Let $\Match$ be a greedy MWM in which $\Bidderp$ is matched to $\Item$.
  Thus $\Bidder$ is not matched in $\Match$.
  Let $\Matchp$ denote the matching $\Match - (\Bidderp, \Item) + (\Bidder, \Item)$.
  Since $\Match$ is an MCMWM of $\Auc$ and $\WMatch{\Matchp} = \WMatch{\Match} - \Offerp + \Offer$, we conclude that $\Offer\leq\Offerp$.
  If $\Offer < \Offerp$, the claim of the lemma follows.
  Assume that $\Offer = \Offerp$.
  In this case, $\Matchp$ is an MCMWM of $\Auc$ since $\WMatch{\Matchp} = \WMatch{\Match}$ and $\Car{\Matchp} = \Car{\Match}$.
  Since $\Match$ is a greedy MWM of $\Auc$ and $\MPriority{\Matchp} = \MPriority{\Match} - \Prip + \Pri$, we conclude that $\Pri \leq \Prip$.
  If $\Pri = \Prip$ then $\Matchp$ is a greedy MWM of $\Auc$ that matches $\Bidder$, a contradiction.
  Hence $\Pri < \Prip$, as required.
\end{proof}

\subsection{Extending a UAP}
\label{sec:uap-extend}

Let $\Auc = (\Bidders, \Items)$ be a UAP and let $\Bidder$ be a bidder such that $\BId{\Bidder}$ is not equal to the ID of any bidder in $\Bidders$.
Then we define $\Auc + \Bidder$ as the UAP $(\Bidders + \Bidder, \Items)$.
For any UAPs $\Auc = (\Bidders, \Items)$ and $\Aucp = (\Biddersp, \Itemsp)$, we say that \emph{$\Aucp$ extends $\Auc$} if $\Bidders \subseteq \Biddersp$ and $\Items = \Itemsp$.

\begin{lemma}
  \label{lem:uap-one-way-street}
  Let $\Auc = (\Bidders, \Items)$ be a UAP, let $\Bidder$ be a bidder in $\Bidders$ that is not matched in any greedy MWM of $\Auc$, and let $\Aucp = (\Biddersp, \Items)$ be a UAP that extends $\Auc$.
  Then $\Bidder$ is not matched in any greedy MWM of $\Aucp$.
\end{lemma}

\begin{proof}
  In what follows, we derive a contradiction by proving that $\Bidder$ is matched in a greedy MWM $\MatchSub{1}$ of $\Aucp$.
  We need to prove that $\Bidder$ is matched in some greedy MWM of $\Auc$.
  Let $\MatchSub{0}$ denote a greedy MWM of $\Auc$.
  If $\Bidder$ is matched in $\MatchSub{0}$, we are done, so assume that $\Bidder$ is not matched in $\MatchSub{0}$.
  Thus $\MatchSub{0} \oplus \MatchSub{1}$ contains a unique path $\Path$ with $\Bidder$ as an endpoint.
  The edges of $\Path$ alternate between $\MatchSub{0}$ and $\MatchSub{1}$.
  Let $X_0$ denote the edges of $\Path$ that belong to $\MatchSub{0}$, and let $X_1$ denote the edges of $\Path$ that belong to $\MatchSub{1}$.

  Since $\Bidder$ is matched in $\MatchSub{1}$ and not in $\MatchSub{0}$, the other endpoint of $\Path$ is either an item, or it is a bidder that is matched in $\MatchSub{0}$ and not in $\MatchSub{1}$.
  Either way, we deduce that all of the vertices on $\Path$ belong to $\Auc$.
  Thus $\MatchSubp{0} = \MatchSub{0} \oplus \Path = (\MatchSub{0} \cup X_1) \setminus X_0$ is a matching in $\Auc$.
  Since $\MatchSub{0}$ is an MWM of $\Auc$ and $\MatchSubp{0}$ is a matching of $\Auc$, we deduce that $\WMatch{\MatchSub{0}} \geq \WMatch{\MatchSubp{0}}$ and hence that $\WEdges{X_0} \geq \WEdges{X_1}$.
  Since all of the vertices on $\Path$ belong to $\Aucp$, we conclude that $\MatchSubp{1} = \MatchSub{1} \oplus \Path = (\MatchSub{1} \cup X_0) \setminus X_1$ is a matching in $\Aucp$.
  Since $\MatchSub{1}$ is an MWM of $\Aucp$ and $\MatchSubp{1}$ is a matching of $\Aucp$, we deduce that $\WMatch{\MatchSub{1}} \geq \WMatch{\MatchSubp{1}}$ and hence that $\WEdges{X_1} \geq \WEdges{X_0}$.
  Thus $\WEdges{X_0} = \WEdges{X_1}$, and we conclude that $\MatchSubp{0}$ is an MWM of $\Auc$.

  Since $\Bidder$ is matched in $\MatchSub{1}$ and not in $\MatchSub{0}$, we deduce that $\Car{X_1} \geq \Car{X_0}$ and hence that $\Car{\MatchSubp{0}} \geq \Car{\MatchSub{0}}$.
  Since $\MatchSub{0}$ is a greedy MWM of $\Auc$, we know that $\MatchSub{0}$ is an MCMWM of $\Auc$, and hence that $\Car{\MatchSub{0}} \geq \Car{\MatchSubp{0}}$.
  Thus $\Car{\MatchSub{0}} = \Car{\MatchSubp{0}}$ and hence $\Car{X_0} = \Car{X_1}$ and $\MatchSubp{0}$ is an MCMWM of $\Auc$.
  Since $\Car{X_0} = \Car{X_1}$, the other endpoint of $\Path$ is a bidder $\Bidder'$ that is matched in $\MatchSub{0}$ and not in $\MatchSub{1}$.
  Since $\MatchSub{0}$ is a greedy MWM of $\Auc$ and $\MatchSubp{0}$ is an MCMWM of $\Auc$, we deduce that $\MPriority{\MatchSub{0}} \geq \MPriority{\MatchSubp{0}}$ and hence that $\BPri{\Bidderp} \geq \BPri{\Bidder}$.

  Since $\Car{X_0} = \Car{X_1}$ and $\WEdges{X_0} = \WEdges{X_1}$, we deduce that $\MatchSubp{1}$ is an MCMWM of $\Aucp$.
  Since $\MatchSub{1}$ is a greedy MWM of $\Aucp$ and $\MatchSubp{1}$ is an MCMWM of $\Aucp$, we deduce that $\MPriority{\MatchSub{1}} \geq \MPriority{\MatchSubp{1}}$ and hence that $\BPri{\Bidder} \geq \BPri{\Bidderp}$.
  Since we argued above that $\BPri{\Bidderp} \geq \BPri{\Bidder}$, we conclude that $\BPri{\Bidder} = \BPri{\Bidderp}$, and hence that $\MatchSubp{0}$ is a greedy MWM of $\Auc$.
  This completes the proof, since $\Bidder$ is matched in $\MatchSubp{0}$.
\end{proof}

\subsection{Finding a Greedy MWM}
\label{sec:inc-hungarian-step}

In this section, we briefly discuss how to efficiently compute a greedy MWM of a UAP via a slight modification of the classic Hungarian method for the assignment problem~\citep{Kuhn55}.
In the
(maximization version of the) assignment problem, we are given a set
of $n$ agents, a set of $n$ tasks, and a weight for each agent-task
pair, and our objective is to find a perfect matching (i.e., every
agent and task is required to be matched) of maximum total weight.
The Hungarian method for the assignment problem proceeds as follows:
a set of dual variables, namely a ``price'' for each task, and a possibly incomplete matching are maintained;
an arbitrary unmatched agent $\Bidder$ is chosen and a shortest augmenting path from $\Bidder$ to an unmatched task is computed using ``residual costs'' as the edge weights;
an augmentation is performed along the path to update the matching, and the dual variables are adjusted in order to maintain complementary slackness;
the process repeats until a perfect matching is found.

\newcommand{\ItemD}{\ItemSub{0}}

Within our UAP setting, the set of bidders can be larger than the set of items, and some bidder-item pairs may not be matchable, i.e., the associated bipartite graph is not necessarily complete.
In this setting, we can use an ``incremental'' version of the Hungarian method to find an (not necessarily greedy) MWM of a given UAP $\Auc = (\Bidders, \Items)$ as follows.
For the purpose of simplifying the presentation of our method, we enlarge the set of items by adding a dummy item $\ItemD$ such that $\ItemD$ is connected to each bidder $\Bidder$ with an edge of weight $\WPair{\Bidder}{\ItemD} = 0$ and we always maintain $\ItemD$ in the residual graph with a price of $0$.
We start with the empty matching $\Match$.
Then, for each bidder $\Bidder$ in $\Bidders$ (in arbitrary order), we process $\Bidder$ via an ``incremental Hungarian step'' as follows:
let $\Biddersp$ denote the set of bidders that are matched by $\Match$;
let $\Itemsp$ denote the set of items that are not matched by $\Match$;
find the shortest paths from $\Bidder$ to each item $\Item$ in $\Itemsp + \ItemD$ in the residual graph;
let $W$ denote the minimum path weight among these shortest paths;
choose a path $\Path$ that is either (1) a shortest path of weight $W$ from $\Bidder$ to an item $\Item$ in $\Itemsp$, or (2) a shortest path from $\Bidder$ to a bidder $\Bidderp$ in $\Biddersp + \Bidder$ such that extending $\Path$ with the edge $(\Bidderp, \ItemD)$ yields a shortest path of weight $W$ from $\Bidder$ to $\ItemD$;
augment $\Match$ along $\Path$;
adjust the prices in order to maintain complementary slackness;
update the residual graph.
The algorithm terminates when every non-reserve bidder has been processed.
The algorithm performs $\Car{\Bidders}$ incremental Hungarian steps and each incremental Hungarian step can be implemented in $\Oh{\Car{\Items} \log \Car{\Items} + m}$ time by utilizing Fibonacci heaps~\citep{FT87}, where $m$ denotes the number of edges in the residual graph, which is $\Oh{\Car{\Items}^2}$.

In order to find a greedy MWM, we slightly modify the implementation described in the previous paragraph.
Lemmas~\ref{lem:uap-grow} and~\ref{lem:uap-swap} established below imply that choosing the path $\Path$ in the following way results in a greedy MWM:
if a path of type (1) exists, we arbitrarily choose such a path;
if no path of type (1) exists, then we identify the nonempty set $\Bidderspp$ of all bidders $\Bidderp$ such that a path of type (2) exists, and we choose a shortest path $\Path$ that terminates at a minimum priority bidder in $\Bidderspp$.
It is easy to see that the described modification does not increase the asymptotic time complexity of the algorithm.
In the remainder of this section, we establish Lemmas~\ref{lem:uap-grow}, \ref{lem:uap-swap}, and \ref{lem:greedy-mwm-card};
Lemma~\ref{lem:greedy-mwm-card} is used in Section~\ref{sec:proof-iuap-threshold-proof} to prove Lemma~\ref{lem:iuap-priorities}.
We start with some useful definitions.

\newcommand{\Digraph}[3]{\Id{digraph}(#1, #2, #3)}

Let $\Auc = (\Bidders, \Items)$ and $\Aucp = \Auc + \Bidder$ be UAPs, and let $\Match$ be an MWM of $\Auc$.
We define $\Digraph{\Auc}{\Bidder}{\Match}$ as the edge-weighted digraph that may be obtained by modifying the subgraph of $\Auc$ induced by the set of vertices $(\MatchedBidders{\Match} + \Bidder) \cup \Items$ as follows:
for each edge that belongs to $\Match$, we direct the edge from item to bidder and leave the weight unchanged;
for each edge that does not belong to $\Match$, we direct the edge from bidder to item and negate the weight.

% Graphs
\newcommand{\Graph}{G}

\begin{lemma}
  \label{lem:uap-cycle}
  Let $\Auc = (\Bidders, \Items)$ and $\Aucp = \Auc + \Bidder$ be UAPs, let $\Match$ be an MWM of $\Auc$, and let $\Graph$ denote $\Digraph{\Auc}{\Bidder}{\Match}$.
  Then $\Graph$ does not contain any negative-weight cycles.
\end{lemma}

\begin{proof}
  Such a cycle could not involve $\Bidder$ (since $\Bidder$ only has outgoing edges) so it has to be a negative-weight cycle that already existed before $\Bidder$ was added, a contradiction since $\Match$ is an MWM of $\Auc$.
\end{proof}

\newcommand{\Holes}[3]{\Id{holes}(#1, #2, #3)}
\newcommand{\Candidates}[3]{\Id{candidates}(#1, #2, #3)}

Let $\Auc = (\Bidders, \Items)$ and $\Aucp = \Auc + \Bidder$ be UAPs, let $\Match$ be an MWM of $\Auc$, and let $\Graph$ denote $\Digraph{\Auc}{\Bidder}{\Match}$.
We define a set of items $\Holes{\Auc}{\Bidder}{\Match}$, and a set of bidders $\Candidates{\Auc}{\Bidder}{\Formatting{\allowbreak} \Match}$, as follows.
By Lemma~\ref{lem:uap-cycle}, the shortest path distance in $\Graph$ from bidder $\Bidder$ to any vertex reachable from $\Bidder$ is well-defined.
We define $\Holes{\Auc}{\Bidder}{\Match}$ as the set of all items $\Item$ in $\Items$ such that $\Item$ is unmatched in $\Match$ and the weight of a shortest path in $\Graph$ from $\Bidder$ to $\Item$ is $\WAuc{\Auc} - \WAuc{\Aucp}$.
We define $\Candidates{\Auc}{\Bidder}{\Match}$ as the set of all bidders $\Bidderp$ such that the weight of a shortest path in $\Graph$ from $\Bidder$ to $\Bidderp$ is equal to $\WAuc{\Auc} - \WAuc{\Auc'}$.

\newcommand{\Augment}[4]{\Id{augment}(#1, #2, #3, #4)}

Let $\Auc = (\Bidders, \Items)$ and $\Aucp = \Auc + \Bidder$ be UAPs, let $\Match$ be an MWM of $\Auc$, and let $\Path$ be a directed path in $\Digraph{\Auc}{\Bidder}{\Match}$ that starts at $\Bidder$, has weight $\WAuc{\Auc} - \WAuc{\Auc'}$, and terminates at either an item in $\Holes{\Auc}{\Bidder}{\Match}$ or a bidder in $\Candidates{\Auc}{\Bidder}{\Match}$.
(Note that $\Path$ could be a path of length zero from $\Bidder$ to $\Bidder$.)
Let $X$ denote the edges in $\Match$ that correspond to item-to-bidder edges in $\Path$, and let $Y$ denote the edges in $\Aucp$ that correspond to bidder-to-item edges in $\Path$.
It is easy to see that the set of edges $(\Match \setminus X) \cup Y$ is an MWM of $\Aucp$.
We define this MWM of $\Aucp$ as $\Augment{\Auc}{\Bidder}{\Match}{\Path}$.

\begin{lemma}
  \label{lem:uap-xor}
  Let $\Auc = (\Bidders, \Items)$ be a UAP, let $\Match$ be a greedy MWM of $\Auc$, let $\Bidder$ be a bidder that does not belong to $\Bidders$, let $\Aucp$ denote the UAP $(\Bidders + \Bidder, \Items)$, and let $\Matchp$ denote a greedy MWM of $\Aucp$ minimizing $|\Match \oplus \Matchp|$.
  Then $\Digraph{\Auc}{\Bidder}{\Match}$ contains a directed path $\Path$ satisfying the following conditions:
  $\Path$ has weight $\WAuc{\Auc} - \WAuc{\Aucp}$;
  $\Path$ starts at $\Bidder$;
  the bidder-to-item edges in $\Path$ correspond to the edges in $\Matchp \setminus \Match$;
  the item-to-bidder edges in $\Path$ correspond to the edges in $\Match \setminus \Matchp$;
  if $\Holes{\Auc}{\Bidder}{\Match}$ is nonempty, then $\Path$ terminates at an item in $\Holes{\Auc}{\Bidder}{\Match}$;
  if $\Holes{\Auc}{\Bidder}{\Match}$ is empty, then $\Path$ terminates at a minimum-priority bidder in $\Candidates{\Auc}{\Bidder}{\Match}$.
\end{lemma}

\begin{proof}
  The edges of $\Match \oplus \Matchp$ form a collection $\PathCycleColl$ of disjoint cycles and paths of positive length.

  We begin by arguing that $\PathCycleColl$ does not contain any cycles.
  Suppose there is a cycle $\Cycle$ in $\PathCycleColl$.
  Let $X$ denote the edges of $\Cycle$ that belong to $\Match \setminus \Matchp$, and let $Y$ denote the edges of $\Cycle$ that belong to $\Matchp \setminus \Match$.
  Let $\Matchpp$ denote $(\Match \cup Y) \setminus X$, which is a matching in $\Auc$ since $\Bidder$ is unmatched in $\Match$ and hence does not belong to $\Cycle$.
  Since $\Match$ is an MWM of $\Auc$ and $\WMatch{\Matchpp} = \WMatch{\Match} + \WEdges{Y} - \WEdges{X}$, we conclude that $\WEdges{X} \geq \WEdges{Y}$.
  Let $\Matchppp$ denote $(\Matchp \cup X) \setminus Y$, which is a matching in $\Aucp$.
  Since $\Matchp$ is an MWM of $\Aucp$ and $\WMatch{\Matchppp} = \WMatch{\Matchp} + \WEdges{Y} - \WEdges{X}$, we conclude that $\WEdges{X} \leq \WEdges{Y}$.
  Thus $\WEdges{X}= \WEdges{Y}$ and hence $\WMatch{\Matchppp} = \WMatch{\Matchp}$, implying that $\Matchppp$ is an MWM of $\Aucp$.
  Moreover, since $\Matchppp$ matches the same set of bidders as $\Matchp$, we find that $\Matchppp$ is a greedy MWM of $\Aucp$.
  This contradicts the definition of $\Matchp$ since $\Car{\Match \oplus \Matchppp} < \Car{\Match \oplus \Matchp}$.

  Next we argue that if $\PathOther$ is a path in $\PathCycleColl$, then $\Bidder$ is an endpoint of $\PathOther$.
  Suppose there is a path $\PathOther$ in $\PathCycleColl$ such that $\Bidder$ is not an endpoint of $\PathOther$.
  Thus $\Bidder$ does not appear on $\PathOther$ since $\Bidder$ is unmatched in $\Match$.
  Let $X$ denote the edges of $\PathOther$ that belong to $\Match \setminus \Matchp$, and let $Y$ denote the edges of $\PathOther$ that
  belong to $\Matchp \setminus \Match$.
  Let $\Matchpp$ denote $(\Match \cup Y) \setminus X$, which is a matching in $\Auc$ since $\Bidder$ does not belong to $\PathOther$.
  Since $\Match$ is an MWM of $\Auc$ and $\WMatch{\Matchpp} = \WMatch{\Match} + \WEdges{Y} - \WEdges{X}$, we conclude that $\WEdges{X} \geq \WEdges{Y}$.
  Let $\Matchppp$ denote $(\Matchp \cup X) \setminus Y$, which is a matching in $\Aucp$.
  Since $\Matchp$ is an MWM of $\Aucp$ and $\WMatch{\Matchppp} = \WMatch{\Matchp} + \WEdges{Y} - \WEdges{X}$, we conclude that $\WEdges{X} \leq \WEdges{Y}$.
  Thus $\WEdges{X}= \WEdges{Y}$ and hence $\WMatch{\Matchpp} = \WMatch{\Match}$ and $\WMatch{\Matchppp} = \WMatch{\Matchp}$, implying that $\Matchpp$ is an MWM of $\Auc$ and $\Matchppp$ is an MWM of $\Aucp$.
  Since $\Match$ is a greedy MWM and hence an MCMWM of $\Auc$, the set of bidders matched by $\Match$ is not properly contained in the set of bidders matched by $\Matchpp$;
  we conclude that $\Car{X} \geq \Car{Y}$.
  Since $\Matchp$ is a greedy MWM and hence an MCMWM of $\Aucp$, the set of bidders matched by $\Matchp$ is not properly contained in the set of bidders matched by $\Matchppp$;
  we conclude that $\Car{X} \leq \Car{Y}$.
  Thus $\Car{X} = \Car{Y}$, so the length of path $\PathOther$ is even.
  We consider two cases.

  Case 1: The endpoints of $\PathOther$ are items.
  In this case, $\Matchp$ and $\Matchppp$ match the same set of bidders, and hence $\Matchppp$ is a greedy MWM of $\Aucp$.
  This contradicts the definition of $\Matchp$, since $\PathOther$ has positive length and hence $\Car{\Match \oplus \Matchppp} < \Car{\Match \oplus \Matchp}$.

  Case 2: The endpoints of $\PathOther$ are bidders.
  Since $\PathOther$ has positive length, one endpoint, call it $\BidderSub{0}$, is matched in $\Match$ but not in $\Matchp$, and the other endpoint, call it $\BidderSub{1}$, is matched in $\Matchp$ but not in $\Match$.
  Since $\Match$ is a greedy MWM of $\Auc$ and $\Matchpp$ is an MWM of $\Auc$, we deduce that $\BPri{\BidderSub{0}} \geq \BPri{\BidderSub{1}}$.
  Since $\Matchp$ is a greedy MWM of $\Aucp$ and $\Matchppp$ is an MWM of $\Aucp$, we deduce that $\BPri{\BidderSub{0}} \leq \BPri{\BidderSub{1}}$.
  Thus $\BPri{\BidderSub{0}} = \BPri{\BidderSub{1}}$.
  It follows that $\MPriority{\Matchppp} = \MPriority{\Matchp}$.
  Hence $\Matchppp$ is a greedy MWM of $\Aucp$.
  This contradicts the definition of $\Matchp$ since $\Car{\Match \oplus \Matchppp} < \Car{\Match \oplus \Matchp}$.

  From the preceding arguments, we deduce that either $\Match = \Matchp$ or $\Match \oplus \Matchp$ corresponds to a positive-length path with $\Bidder$ as an endpoint.
  Equivalently, $\Match \oplus \Matchp$ is the edge set of a path that has $\Bidder$ as an endpoint and may have length zero (i.e., the path may begin and end at $\Bidder$).
  We claim if the edges of this path are directed away from endpoint $\Bidder$, we obtain a directed path $\Path$ satisfying the six conditions stated in the lemma.
  It is easy to see that $\Path$ satisfies the first four of these conditions.
  It remains to establish that $\Path$ satisfies the fifth and sixth conditions.

  For the fifth condition, assume that $\Holes{\Auc}{\Bidder}{\Match}$ is nonempty.
  We need to prove that $\Path$ terminates at an item in $\Holes{\Auc}{\Bidder}{\Match}$.
  Since $\Holes{\Auc}{\Bidder}{\Match}$ is nonempty, we deduce that $\Car{\Matchp} = \Car{\Match} + 1$, and hence that $\Path$ terminates at some item $\Item$.
  Since $\Path$ has weight $\WAuc{\Auc} - \WAuc{\Aucp}$, we deduce that $\Item$ belongs to $\Holes{\Auc}{\Bidder}{\Match}$, as required.

  For the sixth condition, assume that $\Holes{\Auc}{\Bidder}{\Match}$ is empty.
  We need to prove that $\Path$ terminates at a minimum-priority bidder in $\Candidates{\Auc}{\Bidder}{\Match}$.
  Suppose $\Path$ terminates at some item $\Item$.
  Since $\Path$ has weight $\WAuc{\Auc} - \WAuc{\Aucp}$, we deduce that $\Item$ belongs to $\Holes{\Auc}{\Bidder}{\Match}$, a contradiction.
  Thus $\Path$ terminates at some bidder $\Bidderp$.
  Since $\Path$ has weight $\WAuc{\Auc} - \WAuc{\Aucp}$, we deduce that $\Bidderp$ belongs to $\Candidates{\Auc}{\Bidder}{\Match}$.
  If $\Bidderp$ is not a minimum-priority bidder in $\Candidates{\Auc}{\Bidder}{\Match}$, it is easy to argue that $\Matchp$ is not a greedy MWM of $\Aucp$, a contradiction.
  Thus $\Path$ terminates at a minimum-priority bidder in $\Candidates{\Auc}{\Bidder}{\Match}$.
\end{proof}

\begin{lemma}
  \label{lem:uap-grow}
  Let $\Auc = (\Bidders, \Items)$ be a UAP, let $\Match$ be a greedy MWM of $\Auc$, let $\Bidder$ be a bidder that does not belong to $\Bidders$, let $\Aucp$ denote the UAP $(\Bidders + \Bidder, \Items)$, let $\Path$ be a directed path in $\Digraph{\Auc}{\Bidder}{\Match}$ of weight $\WAuc{\Auc} - \WAuc{\Aucp}$ from $\Bidder$ to an item in $\Holes{\Auc}{\Bidder}{\Match}$, and let $\MatchOpt$ denote $\Augment{\Auc}{\Bidder}{\Match}{\Path}$.
  Then $\MatchOpt$ is a greedy MWM of $\Aucp$.
\end{lemma}

\begin{proof}
  The definition of $\Augment{\Auc}{\Bidder}{\Match}{\Path}$ implies that $\MatchOpt$ is an MWM of $\Aucp$.
  Let $\Matchp$ denote a greedy MWM of $\Aucp$ minimizing $\Car{\Match \oplus \Matchp}$.
  Let $\Biddersp$ denote the set of bidders in $\Auc$ matched by $\Match$.
  Since $\Holes{\Auc}{\Bidder}{\Match}$ is nonempty, Lemma~\ref{lem:uap-xor} implies that the set of bidders in $\Aucp$ matched by $\Matchp$ is $\Biddersp + \Bidder$.
  Since $\MatchOpt$ is an MWM of $\Aucp$ that also matches the set of bidders $\Biddersp + \Bidder$, we deduce that $\MatchOpt$ is a greedy MWM of $\Aucp$.
\end{proof}

\begin{lemma}
  \label{lem:uap-swap}
  Let $\Auc = (\Bidders, \Items)$ be a UAP, let $\Match$ be a greedy MWM of $\Auc$, let $\Bidder$ be a bidder that does not belong to $\Bidders$, and let $\Aucp$ denote the UAP $(\Bidders + \Bidder, \Items)$.
  Assume that $\Holes{\Auc}{\Bidder}{\Match}$ is empty.
  Let $\Bidderp$ denote a minimum-priority bidder in $\Candidates{\Auc}{\Bidder}{\Match}$ (which is nonempty by Lemma~\ref{lem:uap-xor}), let $\Path$ be a directed path in $\Digraph{\Auc}{\Bidder}{\Match}$ of weight $\WAuc{\Auc} - \WAuc{\Aucp}$ from $\Bidder$ to $\Bidderp$, and let $\MatchOpt$ denote $\Augment{\Auc}{\Bidder}{\Match}{\Path}$.
  Then $\MatchOpt$ is a greedy MWM of $\Aucp$.
\end{lemma}

\begin{proof}
  The definition of $\Augment{\Auc}{\Bidder}{\Match}{\Path}$ implies that $\MatchOpt$ is an MWM of $\Aucp$.
  Let $\Matchp$ denote a greedy MWM of $\Aucp$ minimizing $\Car{\Match \oplus \Matchp}$.
  Let $\Biddersp$ denote the set of bidders in $\Auc$ matched by $\Match$.
  Since $\Holes{\Auc}{\Bidder}{\Match}$ is empty, Lemma~\ref{lem:uap-xor} implies that the set of bidders in $\Aucp$ matched by $\Matchp$ is $\Biddersp + \Bidder- \Bidderpp$, where $\Bidderpp$ is some minimum-priority bidder in $\Candidates{\Auc}{\Bidder}{\Match}$.
  It is straightforward to check that $\MatchOpt$ has the same weight, cardinality, and priority as $\Matchp$.
  Thus $\MatchOpt$ is a greedy MWM of $\Aucp$, as required.
\end{proof}

\begin{lemma}
  \label{lem:greedy-mwm-card}
  Let $\Auc$ and $\Aucp$ be two UAPs such that $\Aucp$ extends $\Auc$, let $\Match$ be a greedy MWM of $\Auc$, and let $\Matchp$ be a greedy MWM of $\Aucp$.
  Then $\Car{\Matchp} \geq \Car{\Match}$.
\end{lemma}

\begin{proof}
  Immediate from Lemmas~\ref{lem:uap-grow} and~\ref{lem:uap-swap}.
\end{proof}

\subsection{Threshold of an Item}
\label{sec:uap-threshold}

In this section, we define the notion of a ``threshold'' of an item in a UAP.
This lays the groundwork for a corresponding IUAP definition in Section~\ref{sec:iuap-threshold}.
Item thresholds play an important role in our strategyproofness results.

% Priority sets
\newcommand{\Pris}{Z}
% Shortcuts for priority sets
% with primes
\newcommand{\Prisp}{\Pris'}
% with double primes
\newcommand{\Prispp}{\Pris''}

\begin{lemma}
  \label{lem:uap-threshold}
  Let $\Auc = (\Bidders, \Items)$ be a UAP and let $\Item$ be an item in $\Items$.
  Let $\Biddersp$ be the set of bidders $\Bidder$ such that $\Auc + \Bidder$ is a UAP and $\BUBid{\Bidder}$ is of the form $\Set{(\Item, \Offer)}$.
  Then there is a unique pair of reals $(\OfferOpt, \PriOpt)$ such that for any bidder $\Bidder$ in $\Biddersp$, the following conditions hold, where $\Aucp$ denotes $\Auc + \Bidder$, $\Offer$ denotes $\WPair{\Bidder}{\Item}$, and $\Pri$ denotes $\BPri{\Bidder}$:
  (1) if $(\Offer, \Pri) > (\OfferOpt, \PriOpt)$ then $\Bidder$ is matched in every greedy MWM of $\Aucp$;
  (2) if $(\Offer, \Pri) < (\OfferOpt, \PriOpt)$ then $\Bidder$ is not matched in any greedy MWM of $\Aucp$;
  (3) if $(\Offer, \Pri) = (\OfferOpt, \PriOpt)$ then $\Bidder$ is matched in some but not all greedy MWMs of $\Aucp$.
\end{lemma}

\begin{proof}
  Let $\Match$ be a greedy MWM of $\Auc$, let $W$ denote $\WMatch{\Match}$, and let $\Pris$ denote $\MPriority{\Match}$.
  Let $\Matchings$ denote the set of matchings of $\Aucp$ that do not match $\Item$, let $\Matchingsp$ denote the maximum-weight elements of $\Matchings$, let $\Matchingspp$ denote the maximum-cardinality elements of $\Matchingsp$, let $\Matchingsppp$ denote the maximum-priority elements of $\Matchingspp$, and observe that there is a unique pair of reals $(W', \Prisp)$ such that any matching $\Matchp$ in $\Matchingsppp$ has weight $W'$ and priority $\Prisp$.
  It is straightforward to verify that the unique choice of $(\OfferOpt, \PriOpt)$ satisfying the conditions stated in the lemma is $(W - W', \Pris - \Prisp)$.
\end{proof}

% Threshold in a UAP
\newcommand{\ThresholdId}{\Id{threshold}}
\newcommand{\AThreshold}[2]{\ThresholdId(#1, #2)}

For any UAP $\Auc = (\Bidders, \Items)$ and any item $\Item$ in $\Items$, we define the unique pair $(\OfferOpt, \PriOpt)$ of Lemma~\ref{lem:uap-threshold} as $\AThreshold{\Auc}{\Item}$.

\comment{
  The first components of the $\AThreshold{\Auc}{\Item}$ values correspond to the maximum stable price vector.
}

\section{Iterated Unit-Demand Auctions with Priorities}
\label{sec:iuap}

In this section, we formally define the notion of an iterated unit-demand auction with priorities (IUAP).
An IUAP allows the bidders, called ``multibidders'' in this context, to have a sequence of unit-demand bids instead of a single unit-demand bid.
In Section~\ref{sec:unfold-iuap}, we define a mapping from an IUAP to a UAP by describing an algorithm that generalizes the DA algorithm, and we establish Lemma~\ref{lem:iuap-exit} that is useful for analyzing the matching produced by Algorithm~\ref{alg:smiw} of Section~\ref{sec:smiw}.
Lemma~\ref{lem:iuap-exit} is used to establish weak stability (Lemmas~\ref{lem:smiw-valid}, \ref{lem:smiw-rational}, and~\ref{lem:smiw-stable}) and Pareto-optimality (Lemma~\ref{lem:smiw-pareto}).
In Section~\ref{sec:iuap-threshold}, we define the threshold of an item in an IUAP and we establish Lemma~\ref{lem:iuap-threshold-george}, which plays a key role in establishing our strategyproofness results.
We start with some useful definitions.

% Multibidders
\newcommand{\MBidder}{t}
% Shortcuts for multibidders
% with primes
\newcommand{\MBidderp}{\MBidder'}

% Multibidder sets
\newcommand{\MBidders}{T}
% Shortcuts for multibidder sets
% with primes
\newcommand{\MBiddersp}{\MBidders'}

% Bidder sequences
\newcommand{\BidderSeq}{\sigma}
\newcommand{\BidderSeqp}{\BidderSeq'}

% Multibidder functions
\newcommand{\MBPri}[1]{\Id{priority}(#1)}
\newcommand{\MBBidder}[2]{\Id{bidder}(#1, #2)}
\newcommand{\MBBidders}[2]{\Id{bidders}(#1, #2)}
\newcommand{\MBBiddersAll}[1]{\Id{bidders}(#1)}

A \emph{multibidder $\MBidder$ for a set of items $\Items$} is a pair $(\BidderSeq, \Pri)$ where $\Pri$ is a real priority and $\BidderSeq$ is a sequence of bidders for $\Items$ such that all the bidders in $\BidderSeq$ have distinct IDs and a common priority $\Pri$.
We define $\MBPri{\MBidder}$ as $\Pri$.
For any integer $i$ such that $1 \leq i \leq \Car{\BidderSeq}$, we define $\MBBidder{\MBidder}{i}$ as the bidder $\BidderSeq(i)$.
For any integer $i$ such that $0 \leq i \leq \Car{\BidderSeq}$, we define $\MBBidders{\MBidder}{i}$ as $\SetBuild{\MBBidder{\MBidder}{j}}{1 \leq
j \leq i}$.
We define $\MBBiddersAll{\MBidder}$ as $\MBBidders{\MBidder}{\Car{\BidderSeq}}$.

% IUAPs
\newcommand{\IAuc}{B}
% Shortcuts for IUAPs
% with primes
\newcommand{\IAucp}{\IAuc'}
\newcommand{\IAucpp}{\IAuc''}
\newcommand{\IAucppp}{\IAuc'''}
% with subscripts
\newcommand{\IAucSub}[1]{\IAuc_{#1}}

% Bidders of IUAPs
\newcommand{\IABidders}[1]{\Id{bidders}(#1)}

An \emph{iterated UAP (IUAP)} is a pair $\IAuc = (\MBidders, \Items)$ where $\Items$ is a set of items and $\MBidders$ is a set of multibidders for $\Items$.
In addition, for any distinct multibidders $\MBidder$ and $\MBidderp$ in $\MBidders$, the following conditions hold:
$\MBPri{\MBidder} \not= \MBPri{\MBidderp}$;
if $\Bidder$ belongs to $\MBBiddersAll{\MBidder}$ and $\Bidderp$ belongs to $\MBBiddersAll{\MBidderp}$, then $\BId{\Bidder} \not = \BId{\Bidderp}$.
For any IUAP $\IAuc = (\MBidders, \Items)$, we define $\IABidders{\IAuc}$ as the union, over all $\MBidder$ in $\MBidders$, of $\MBBiddersAll{\MBidder}$.

\subsection{Mapping an IUAP to a UAP}
\label{sec:unfold-iuap}

\newcommand{\AlgToUAP}{\textsc{ToUap}}

Having defined the notion of an IUAP, we now describe an algorithm \AlgToUAP\ that maps a given IUAP to a UAP.
Algorithm \AlgToUAP\ generalizes the DA algorithm.
In each iteration of the DA algorithm, a single man is nondeterministically chosen, and this man reveals his next choice.
In each iteration of \AlgToUAP, a single multibidder is nondeterministically chosen, and this multibidder reveals its next bid.
We prove in Lemma~\ref{lem:iuap-confluence} that, like the DA algorithm, algorithm \AlgToUAP\ is confluent:
the output does not depend on the nondeterministic choices made during an execution.
We conclude this section by establishing Lemma~\ref{lem:iuap-exit}, which is useful for analyzing the matching produced by Algorithm~\ref{alg:smiw} in Section~\ref{sec:smiw}.
Lemma~\ref{lem:iuap-exit} is used to establish weak stability (Lemmas~\ref{lem:smiw-valid}, \ref{lem:smiw-rational}, and~\ref{lem:smiw-stable}) and Pareto-optimality (Lemma~\ref{lem:smiw-pareto}).
We start with some useful definitions.

% Prefix predicate
\newcommand{\Prefix}[2]{\Id{prefix}(#1, #2)}

Let $\Auc$ be a UAP $(\Bidders, \Items)$ and let $\IAuc$ be an IUAP $(\MBidders, \Items)$.
The predicate $\Prefix{\Auc}{\IAuc}$ is said to hold if $\Bidders \subseteq \IABidders{\IAuc}$ and for any multibidder $\MBidder$ in $\MBidders$, $\Bidders \cap \MBBiddersAll{\MBidder} = \MBBidders{\MBidder}{i}$ for some $i$.

% Configurations
\newcommand{\Conf}{C}
% Shortcuts for configurations
% with primes
\newcommand{\Confp}{\Conf'}
% with subscripts
\newcommand{\ConfSub}[1]{\Conf_{#1}}

A \emph{configuration $\Conf$} is a pair $(\Auc, \IAuc)$ where $\Auc$ is a UAP, $\IAuc$ is an IUAP, and $\Prefix{\Auc}{\IAuc}$ holds.

% Multibidder of a bidder
\newcommand{\MBidderOfB}[2]{\Id{multibidder}(#1, #2)}

Let $\Conf = (\Auc, \IAuc)$ be a configuration, where $\Auc = (\Bidders, \Items)$ and $\IAuc = (\MBidders, \Items)$, and let $\Bidder$ be a bidder in $\Bidders$.
Then we define $\MBidderOfB{\Conf}{\Bidder}$ as the unique multibidder $\MBidder$ in $\MBidders$ such that $\Bidder$ belongs to $\MBBiddersAll{\MBidder}$.

% Bidders of a multibidder
\newcommand{\BiddersOfMB}[2]{\Id{bidders}(#1, #2)}

Let $\Conf = (\Auc, \IAuc)$ be a configuration where $\Auc = (\Bidders, \Items)$ and $\IAuc = (\MBidders, \Items)$.
For any $\MBidder$ in $\MBidders$, we define $\BiddersOfMB{\Conf}{\MBidder}$ as $\SetBuild{\Bidder \in \Bidders}{\MBidderOfB{\Conf}{\Bidder} = \MBidder}$.

% Ready bidders in a configuration
\newcommand{\CReady}[1]{\Id{ready}(#1)}

Let $\Conf = (\Auc, \IAuc)$ be a configuration where $\IAuc = (\MBidders, \Items)$.
We define $\CReady{\Conf}$ as the set of all bidders $\Bidder$ in $\IABidders{\IAuc}$ such that $\Greedy{\Auc}{\BPri{\Bidder}} = 0$ and $\Bidder = \MBBidder{\MBidder}{\Car{\BiddersOfMB{\Conf}{\MBidder}} + 1}$ where $\MBidder = \MBidderOfB{\Conf}{\Bidder}$.

\begin{algorithm}
\caption{\AlgToUAP$(\IAuc)$}
\label{alg:to-uap}
\begin{algorithmic}[1]
  \Require An IUAP $\IAuc = (\MBidders, \Items)$
  \State $\Auc \gets (\emptyset, \Items)$
  \State $\Conf \gets (\Auc, \IAuc)$
  \While{$\CReady{\Conf}$ is nonempty}
    \State $\Auc \gets \Auc +$ an arbitrary bidder in $\CReady{\Conf}$\label{line:BLAP}
    \State $\Conf \gets (\Auc, \IAuc)$
  \EndWhile
  \State \Return $\Auc$
\end{algorithmic}
\end{algorithm}

Our algorithm for mapping an IUAP to a UAP is Algorithm~\ref{alg:to-uap}.
The input is an IUAP $\IAuc$ and the output is a UAP $\Auc$ such that $\Prefix{\Auc}{\IAuc}$ holds.
The algorithm starts with the UAP consisting of all the items in $\Items$ but no bidders.
At this point, no bidder of any multibidder is ``revealed''.
Then, the algorithm iteratively and nondeterministically chooses a ``ready'' bidder and ``reveals'' it by adding it to the UAP that is maintained in the program variable $\Auc$.
A bidder $\Bidder$ associated with some multibidder $\MBidder = (\BidderSeq, \Pri)$ is ready if $\Bidder$ is not revealed and for each bidder $\Bidderp$ that precedes $\Bidder$ in $\BidderSeq$, $\Bidderp$ is revealed and is not matched in any greedy MWM of $\Auc$.
It is easy to verify that the predicate $\Prefix{\Auc}{\IAuc}$ is an invariant of the algorithm loop:
if a bidder $\Bidder$ belonging to a multibidder $\MBidder$ is to be revealed at an iteration, and $\Bidders \cap \MBBiddersAll{\MBidder} = \MBBidders{\MBidder}{i}$ for some integer $i$ at the beginning of this iteration, then $\Bidders \cap \MBBiddersAll{\MBidder} = \MBBidders{\MBidder}{i+1}$ after revealing $\Bidder$, where $(\Bidders, \Items)$ is the UAP that is maintained by the program variable $\Auc$ at the beginning of the iteration.
No bidder can be revealed more than once since a bidder cannot be ready after it has been revealed;
it follows that the algorithm terminates.
We now argue that the output of the algorithm is uniquely determined (Lemma~\ref{lem:iuap-confluence}), even though the bidder that is revealed in each iteration is chosen nondeterministically.

\comment{
  In the next definition, we write ``some greedy MWM'' because we have yet to establish Lemma~\ref{lem:iuap-unique}.
}

% Tail predicate of a configuration
\newcommand{\CTailId}{\Id{tail}}
\newcommand{\CTail}[1]{\CTailId(#1)}

For any configuration $\Conf = (\Auc, \IAuc)$, we define the predicate $\CTail{\Conf}$ to hold if for any bidder $\Bidder$ that is matched in some greedy MWM of $\Auc$, we have $\Bidder = \MBBidder{\MBidder}{\Car{\BiddersOfMB{\Conf}{\MBidder}}}$ where $\MBidder$ denotes $\MBidderOfB{\Conf}{\Bidder}$.

\begin{lemma}
  \label{lem:iuap-one}
  Let $\Conf = (\Auc, \IAuc)$ be a configuration where $\IAuc = (\MBidders, \Items)$ and assume that $\CTail{\Conf}$ holds.
  Then $\Greedy{\Auc}{\MBPri{\MBidder}} \leq 1$ for each $\MBidder$ in $\MBidders$.
\end{lemma}

\begin{proof}
  The claim of the lemma easily follows from the definition of $\CTail{\Conf}$.
\end{proof}

\begin{lemma}
  \label{lem:iuap-tail}
  The predicate $\CTail{\Conf}$ is an invariant of the Algorithm~\ref{alg:to-uap} loop.
\end{lemma}

\begin{proof}
  It is easy to see that $\CTail{\Conf}$ holds when the loop is first encountered.
  Now consider an iteration of the loop that takes us from configuration $\Conf = (\Auc, \IAuc)$ where $\Auc = (\Bidders, \Items)$ to configuration $\Confp = (\Aucp, \IAuc)$ where $\Aucp = (\Biddersp, \Items)$.
  We need to show that $\CTail{\Confp}$ holds.
  Let $\Bidder$ be a bidder that is matched in some greedy MWM $\Matchp$ of $\Aucp$.
  Let $\BidderOpt$ denote the bidder that is added to $\Auc$ in line~\ref{line:BLAP}, and consider the following three cases.

  Case 1: $\Bidder = \BidderOpt$.
  Let $\MBidder$ denote $\MBidderOfB{\Confp}{\BidderOpt}$.
  In this case, $\Car{\BiddersOfMB{\Conf}{\MBidder}} + 1 = \Car{\BiddersOfMB{\Confp}{\MBidder}}$, so $\BidderOpt = \MBBidder{\MBidder}{\Car{\BiddersOfMB{\Confp}{\MBidder}}}$, as required.

  Case 2: $\Bidder \not= \BidderOpt$ and $\BPri{\Bidder} \not= \BPri{\BidderOpt}$.
  Since $\Biddersp$ contains $\Bidders$, Lemma~\ref{lem:uap-one-way-street} implies that $\Bidder$ is matched in some greedy MWM of $\Auc$.
  Since $\Conf$ is a configuration and $\CTail{\Conf}$ holds, we deduce that $\Bidder = \MBBidder{\MBidder}{\Car{\BiddersOfMB{\Conf}{\MBidder}}}$ where $\MBidder$ denotes $\MBidderOfB{\Conf}{\Bidder}$.
  Since $\MBidderOfB{\Confp}{\Bidder} = \MBidderOfB{\Conf}{\Bidder}$ and $\BiddersOfMB{\Confp}{\MBidder} = \BiddersOfMB{\Conf}{\MBidder}$, we conclude that $\Bidder$ is equal to $\MBBidder{\MBidder}{\Formatting{\allowbreak} \Car{\BiddersOfMB{\Confp}{\MBidder}}}$ where $\MBidder$ denotes $\MBidderOfB{\Confp}{\Bidder}$, as required.

  Case 3: $\Bidder \not= \BidderOpt$ and $\BPri{\Bidder} = \BPri{\BidderOpt}$.
  Since $\BidderOpt$ belongs to $\CReady{\Conf}$, we know that $\Greedy{\Auc}{\BPri{\Bidder}} = 0$.
  Also, since $\Bidder$ is not $\BidderOpt$, $\Bidder$ belongs to $\Bidders$ and we conclude that $\Bidder$ is not matched in any greedy MWM of $\Auc$.
  Since $\Biddersp$ contains $\Bidders$, Lemma~\ref{lem:uap-one-way-street} implies that $\Bidder$ is not matched in any greedy MWM of $\Aucp$, a contradiction.
\end{proof}

\begin{lemma}
  \label{lem:iuap-unique}
  Let $\Conf = (\Auc, \IAuc)$ be a configuration such that $\CTail{\Conf}$ holds.
  Then $\Unique{\Auc}$ holds.
\end{lemma}

\begin{proof}
  Let $\Match$ and $\Matchp$ be greedy MWMs of $\Auc$, and let $\Bidder$ be a bidder in $\MatchedBidders{\Match}$.
  To establish the lemma, it is sufficient to prove that $\Bidder$ belongs to $\MatchedBidders{\Matchp}$.
  Let $\MBidder$ denote $\MBidderOfB{\Conf}{\Bidder}$ and let $\Pri$ denote $\MBPri{\MBidder}$.
  Since $\CTail{\Conf}$ holds, we know that $\Bidder = \MBBidder{\MBidder}{\Car{\BiddersOfMB{\Conf}{\MBidder}}}$.
  Since $\Bidder$ is matched by $\Match$ and since $\CTail{\Conf}$ holds, Lemma~\ref{lem:iuap-one} implies that $\Greedy{\Auc}{\Pri} = 1$.
  Thus Lemma~\ref{lem:uap-distribution} implies that $\Matchp$ matches one priority-$\Pri$ bidder.
  Since $\CTail{\Conf}$ holds, this bidder is $\Bidder$.
\end{proof}

\begin{lemma}
  \label{lem:iuap-confluence}
  Let $\IAuc = (\MBidders, \Items)$ be an IUAP\@.
  Then all executions of Algorithm~\ref{alg:to-uap} on input $\IAuc$ produce the same output.
\end{lemma}

% Executions
\newcommand{\Exec}{X}
% Shortcuts for executions
% with subscripts
\newcommand{\ExecSub}[1]{\Exec_{#1}}

\begin{proof}
  Suppose not, and let $\ExecSub{1}$ and $\ExecSub{2}$ denote two executions of Algorithm~\ref{alg:to-uap} on input $\IAuc$ that produce distinct output UAPs $\AucSub{1} = (\BiddersSub{1}, \Items)$ and $\AucSub{2} = (\BiddersSub{2}, \Items)$.
  Without loss of generality, assume that $\Car{\BiddersSub{1}} \geq \Car{\BiddersSub{2}}$.
  Then there is a first iteration of execution $\ExecSub{1}$ in which the bidder added to $\Auc$ in line~\ref{line:BLAP} belongs to $\BiddersSub{1} \setminus \BiddersSub{2}$;
  let $\Bidderp$ denote this bidder.
  Let $\Confp = (\Aucp, \IAuc)$ where $\Aucp = (\Biddersp, \Items)$ denote the configuration in program variable $\Conf$ at the start of this iteration, and let $\MBidderp$ denote $\MBidderOfB{\Confp}{\Bidderp}$.
  Let $i$ be the integer such that $\Bidderp = \MBBidder{\MBidderp}{i}$.
  We know that $i > 1$ because it is easy to see that $\BiddersSub{2}$ contains $\MBBidder{\MBidderp}{1}$.
  Let $\Bidderpp$ denote $\MBBidder{\MBidderp}{i-1}$.
  Since $\Bidderp$ belongs to $\CReady{\Confp}$, Lemmas~\ref{lem:iuap-tail} and~\ref{lem:iuap-unique} imply that $\Bidderpp$ is not matched in any greedy MWM of $\Aucp$.
  Since $\Biddersp$ is contained in $\BiddersSub{2}$, Lemma~\ref{lem:uap-one-way-street} implies that $\Bidderpp$ is not matched in any greedy MWM of $\AucSub{2}$.
  Let $\ConfSub{2} = (\AucSub{2}, \IAuc)$ denote the final configuration of execution $\ExecSub{2}$;
  thus $\CReady{\ConfSub{2}}$ is empty and $\Car{\BiddersOfMB{\ConfSub{2}}{\MBidderp}} = i - 1$.
  By Lemma~\ref{lem:iuap-tail}, we conclude that $\Greedy{\AucSub{2}}{\MBPri{\MBidderp}} = 0$, and hence that $\Bidderpp$ is contained in $\CReady{\ConfSub{2}}$, a contradiction.
\end{proof}

\newcommand{\UAPId}{\Id{uap}}
\newcommand{\UAP}[1]{\UAPId(#1)}

For any IUAP $\IAuc$, we define $\UAP{\IAuc}$ as the unique (by Lemma~\ref{lem:iuap-confluence}) UAP returned by any execution of Algorithm~\ref{alg:to-uap} on input $\IAuc$.

We can use the modified incremental Hungarian step of Section~\ref{sec:inc-hungarian-step} in each iteration of the loop of Algorithm~\ref{alg:to-uap} to maintain UAP $\Auc$, and a greedy MWM of $\Auc$, as follows:
we maintain dual variables (a price for each item) and a residual graph;
the initial greedy MWM is the empty matching;
when a bidder $\Bidder$ is added to $\Auc$ at line~\ref{line:BLAP}, we perform an incremental Hungarian step to process $\Bidder$ to update the greedy MWM, the prices, and the residual graph.
Since we maintain a greedy MWM of $\Auc$ at each iteration of the loop, it is easy to see that identifying a bidder in $\CReady{\Conf}$ (or determining that this set is empty)  takes $\Oh{\Car{\Items}}$ time.
Thus the whole algorithm can be implemented in $\Oh{\Car{\IABidders{\IAuc}} \cdot \Car{\Items}^2}$ time.

We now present a lemma that is used in Section~\ref{sec:smiw} to establish weak stability (Lemmas~\ref{lem:smiw-valid}, \ref{lem:smiw-rational}, and~\ref{lem:smiw-stable}) and Pareto-optimality (Lemma~\ref{lem:smiw-pareto}).

\begin{lemma}
  \label{lem:iuap-exit}
  Let $\IAuc = (\MBidders, \Items)$ be an IUAP,
  let $(\BidderSeq, \Pri)$ be a multibidder that belongs to $\MBidders$,
  let $\UAP{\IAuc}$ be $(\Bidders, \Items)$,
  and let $\Match$ be a greedy MWM of the UAP $(\Bidders, \Items)$.
  Then the following claims hold:
  (1) if $\BidderSeq(k)$ is matched in $\Match$ for some $k$, then $\BidderSeq(k') \in \Bidders$ if and only if $1 \leq k' \leq k$;
  (2) if $\BidderSeq(k)$ is not matched in $\Match$ for any $k$, then $\BidderSeq(k) \in \Bidders$ for $1 \leq k \leq \Car{\BidderSeq}$.
\end{lemma}

\begin{proof}
Since $\Prefix{\Auc}{\IAuc}$ and $\CTail{\Conf}$ hold at the end of
Algorithm~\ref{alg:to-uap} by Lemma~\ref{lem:iuap-tail}, the first claim follows.
Since $\CReady{\Conf}$ is empty at the end of Algorithm~\ref{alg:to-uap},
the second claim follows.
\end{proof}

\subsection{Threshold of an Item}
\label{sec:iuap-threshold}

In this section, we define the threshold of an item in an IUAP and we establish Lemma~\ref{lem:iuap-threshold-george}, which plays a key role in establishing our strategyproofness results.
We start with some useful definitions.

% Winners of an IUAP
\newcommand{\IAWinners}[1]{\Id{winners}(#1)}
% Losers of an IUAP
\newcommand{\IALosers}[1]{\Id{losers}(#1)}

For any IUAP $\IAuc$, Lemmas~\ref{lem:iuap-tail} and~\ref{lem:iuap-unique} imply that $\Unique{\UAP{\IAuc}}$ holds, and thus that every greedy MWM of $\UAP{\IAuc}$ matches the same set of bidders.
We define this set of matched bidders as $\IAWinners{\IAuc}$.
For any IUAP $\IAuc$, we define $\IALosers{\IAuc}$ as $\Bidders \setminus \IAWinners{\IAuc}$ where $(\Bidders, \Items)$ is $\UAP{\IAuc}$.

Let $\IAuc = (\MBidders, \Items)$ be an IUAP and let $\Bidder = (\BidId, \UBid, \Pri)$ be a bidder for $\Items$.
Then we define the IUAP $\IAuc + \Bidder$ as follows:
if $\MBidders$ contains a multibidder $\MBidder$ of the form $(\BidderSeq, \Pri)$ for some sequence of bidders $\BidderSeq$, then we define $\IAuc + \Bidder$ as $(\MBidders - \MBidder + \MBidderp, \Items)$ where $\MBidderp = (\BidderSeqp, \Pri)$ and $\BidderSeqp$ is the sequence of bidders obtained by appending $\Bidder$ to $\BidderSeq$;
otherwise, we define $\IAuc + \Bidder$ as $(\MBidders + \MBidder, \Items)$ where $\MBidder = (\Seq{\Bidder}, \Pri)$.

\begin{lemma}
  \label{lem:iuap-losers-again}
  Let $\IAuc = (\MBidders, \Items)$ and $\IAucp = \IAuc + \Bidder$ be IUAPs.
  Then $\IALosers{\IAuc} \subseteq \IALosers{\IAucp}$.
\end{lemma}

\begin{proof}
  Let $\Bidderp$ be a bidder in $\IALosers{\IAuc}$.
  Thus $\Bidderp$ is not matched in any greedy MWM of $\UAP{\IAuc}$.
  Using Lemma~\ref{lem:iuap-confluence}, it is easy to see that $\UAP{\IAucp}$ extends $\UAP{\IAuc}$.
  Thus Lemma~\ref{lem:uap-one-way-street} implies that $\Bidderp$ is not matched in any greedy MWM of $\UAP{\IAucp}$, and hence that $\Bidderp$ belongs to $\IALosers{\IAucp}$.
\end{proof}

\begin{lemma}
  \label{lem:iuap-threshold-uap}
  Let $\IAuc = (\MBidders, \Items)$ be an IUAP and let $\Item$ be an item in $\Items$.
  For $i \in \Set{1, 2}$, let $\IAucSub{i} = \IAuc + \BidderSub{i}$ be an IUAP where $\BidderSub{i} = (\BidIdSub{i}, \Set{(\Item, \OfferSub{i})}, \PriSub{i})$.
  Let $\AucSub{1} = (\BiddersSub{1}, \Items)$ denote $\UAP{\IAucSub{1}}$ and let $\AucSub{2} = (\BiddersSub{2}, \Items)$ denote $\UAP{\IAucSub{2}}$.
  Assume that $\BidIdSub{1} \not= \BidIdSub{2}$, $\PriSub{1} \not= \PriSub{2}$, and $\BidderSub{1}$ belongs to $\IAWinners{\IAucSub{1}}$.
  Then the following claims hold:
  if $\BidderSub{2}$ belongs to $\IAWinners{\IAucSub{2}}$ then $\BiddersSub{1} - \BidderSub{1} = \BiddersSub{2} - \BidderSub{2}$;
  if $\BidderSub{2}$ belongs to $\IALosers{\IAucSub{2}}$ then $\BiddersSub{1} - \BidderSub{1}$ contains $\BiddersSub{2} - \BidderSub{2}$.
\end{lemma}

\begin{proof}
  Let $\IAucSub{3}$ denote the IUAP $\IAucSub{1} + \BidderSub{2}$, which is equal to $\IAucSub{2} + \BidderSub{1}$.
  For the first claim, assume that $\BidderSub{2}$ belongs to $\IAWinners{\IAucSub{2}}$.
  Using Lemma~\ref{lem:iuap-confluence}, it is straightforward to argue that $\UAP{\IAucSub{3}}$ is equal to $\AucSub{1} + \BidderSub{2} = (\BiddersSub{1} + \BidderSub{2}, \Items)$ and is also equal to $\AucSub{2} + \BidderSub{1} = (\BiddersSub{2} + \BidderSub{1}, \Items)$.
  Since $\BidderSub{1}$ belongs to $\BiddersSub{1}$ and $\BidderSub{2}$ belongs to $\BiddersSub{2}$, we conclude that $\BiddersSub{1} - \BidderSub{1} = \BiddersSub{2} - \BidderSub{2}$, as required.

  For the second claim, assume that $\BidderSub{2}$ belongs to $\IALosers{\IAucSub{2}}$.
  Suppose $(\OfferSub{1}, \PriSub{1}) < (\OfferSub{2}, \PriSub{2})$.
  Then Lemmas~\ref{lem:uap-in-out} and~\ref{lem:iuap-confluence} imply that $\BidderSub{2}$ belongs to $\IAWinners{\IAucSub{3}}$.
  Since $\BidderSub{2}$ belongs to $\IALosers{\IAucSub{2}}$, Lemma \ref{lem:iuap-losers-again} implies that $\BidderSub{2}$ belongs to $\IALosers{\IAucSub{2} + \BidderSub{1}} = \IALosers{\IAucSub{3}}$, a contradiction.
  Since $\PriSub{1} \not= \PriSub{2}$, we conclude that $(\OfferSub{1}, \PriSub{1}) > (\OfferSub{2}, \PriSub{2})$.
  Then, Lemma~\ref{lem:iuap-confluence} implies that $\UAP{\IAucSub{3}} = \UAP{\IAucSub{1}} + \BidderSub{2} = (\BiddersSub{1} + \BidderSub{2}, \Items)$.
  Since Lemma~\ref{lem:iuap-confluence} also implies that $\UAP{\IAucSub{3}}$ extends $\UAP{\IAucSub{2}}$, it follows that $\BiddersSub{1} + \BidderSub{2}$ contains $\BiddersSub{2}$, and hence that $\BiddersSub{1}$ contains $\BiddersSub{2} - \BidderSub{2}$.
  Since $\BidderSub{1}$ does not belong to $\BiddersSub{2}$, we conclude that $\BiddersSub{1} - \BidderSub{1}$ contains $\BiddersSub{2} - \BidderSub{2}$, as required.
\end{proof}

We are now ready to define the threshold of an item in an IUAP, and to state Lemma~\ref{lem:iuap-threshold-george}.
In Section~\ref{sec:smiw}, Lemma~\ref{lem:iuap-threshold-george} plays an important role in establishing that our SMIW mechanism is strategyproof (Lemma~\ref{lem:smiw-threshold}).
The proof of Lemma~\ref{lem:iuap-threshold-george} is provided in Section~\ref{sec:proof-iuap-threshold-proof}.

\comment{
  In the next definition, we are implicitly using the fact that by choosing $\Offer$ large enough, we can ensure that $\Bidder$ belongs to $\IAWinners{\IAucp}$.
}

% Threshold in an IUAP
\newcommand{\UAPwItem}[2]{\Id{uap}(#1, #2)}
\newcommand{\IAThreshold}[2]{\ThresholdId(#1, #2)}

Let $\IAuc = (\MBidders, \Items)$ be an IUAP and let $\Item$ be an item in $\Items$.
By Lemma~\ref{lem:iuap-threshold-uap}, there is a unique subset $\Bidders$ of $\IABidders{\IAuc}$ such that the following condition holds:
for any IUAP $\IAucp = \IAuc + \Bidder$ where $\Bidder$ is of the form $(\BidId, \Set{(\Item, \Offer)}, \Pri)$ and $\Bidder$ belongs to $\IAWinners{\IAucp}$, $\UAP{\IAucp}$ is equal to $(\Bidders + \Bidder, \Items)$.
We define $\UAPwItem{\IAuc}{\Item}$ as the UAP $(\Bidders, \Items)$, and we define $\IAThreshold{\IAuc}{\Item}$ as $\AThreshold{\UAPwItem{\IAuc}{\Item}}{\Item}$.

\begin{lemma}
  \label{lem:iuap-threshold-george}
  Let $\IAuc = (\MBidders, \Items)$ be an IUAP,
  let $\MBidder = (\BidderSeq, \Pri)$ be a multibidder that belongs to $\MBidders$,
  and let $\IAucp$ denote the IUAP $(\MBidders - \MBidder, \Items)$.
  Suppose that $(\BidderSeq(k), \Item)$ is matched in some greedy MWM of $\UAP{\IAuc}$ for some $k$.
  Then
  \begin{equation}
    \label{eq:iuap-threshold-success}
    (\WPair{\BidderSeq(k)}{\Item}, \Pri) \geq \IAThreshold{\IAucp}{\Item} .
  \end{equation}
  Furthermore, for each $k'$ and $\Itemp$ such that $1 \leq k' < k$ and $\Itemp$ belongs to $\BItems{\BidderSeq(k')}$, we have
  \begin{equation}
    \label{eq:iuap-threshold-failure}
    (\WPair{\BidderSeq(k')}{\Itemp}, \Pri) < \IAThreshold{\IAucp}{\Itemp} .
  \end{equation}
\end{lemma}

\subsubsection{Proof of Lemma~\ref{lem:iuap-threshold-george}}
\label{sec:proof-iuap-threshold-proof}

The purpose of this section is to prove Lemma~\ref{lem:iuap-threshold-george}.
We do so by establishing a stronger result, namely Lemma~\ref{lem:iuap-threshold-summary} below.
We start with a useful definition.

% Priorities of an IUAP
\newcommand{\IAPriorities}[1]{\Id{priorities}(#1)}

For any IUAP $\IAuc$, we define $\IAPriorities{\IAuc}$ as $\SetBuild{\Pri}{\Bidder \in \IAWinners{\IAuc} \text{ and } \BPri{\Bidder} = \Pri}$.

\begin{lemma}
  \label{lem:iuap-priorities}
  Let $\IAuc = (\MBidders, \Items)$ and $\IAucp = \IAuc + \Bidder = (\MBiddersp, \Items)$ be IUAPs, let $\Pris$ denote $\IAPriorities{\IAuc}$, let $\Prisp$ denote $\IAPriorities{\IAucp}$, and let $\Pri$ denote $\BPri{\Bidder}$.
  Then $\Car{\Prisp} \geq \Car{\Pris}$ and $\Prisp \subseteq \Pris + \Pri$.
\end{lemma}

\begin{proof}
  Consider running Algorithm~\ref{alg:to-uap} on input $\IAucp$, where we avoid selecting bidder $\Bidder$ from $\CReady{\Conf}$ unless it is the only bidder in $\CReady{\Conf}$.
  (By Lemma~\ref{lem:iuap-confluence}, the final output is the same regardless of which bidder is selected from $\CReady{\Conf}$ at each iteration.)
  If $\Bidder$ never enters $\CReady{\Conf}$, then $\UAP{\IAucp} = \UAP{\IAuc}$, and so $\Prisp = \Pris$, and the claim of the lemma holds.

  Now suppose that $\Bidder$ enters $\CReady{\Conf}$ at some point.
  Let $\Auc = (\Bidders, \Items)$ denote the UAP at the start of the iteration in which $\Bidder$ is selected from $\CReady{\Conf}$.
  Then $\Auc$ is equal to $\UAP{\IAuc}$, and we deduce that $\UAP{\IAucp}$ extends $\UAP{\IAuc}$.
  Lemma~\ref{lem:iuap-one} implies that every greedy MWM of $\Auc = \UAP{\IAuc}$ (resp., $\UAP{\IAucp}$) matches exactly one bidder of each priority in $\Pris$ (resp., $\Prisp$).
  Then, since $\UAP{\IAucp}$ extends $\UAP{\IAuc}$, Lemma~\ref{lem:greedy-mwm-card} implies that $\Car{\Prisp} \geq \Car{\Pris}$.
  Furthermore, letting $\Biddersp$ denote the set of all bidders $\Bidderp$ in $\IABidders{\IAuc}$ such that $\BPri{\Bidderp}$ does not belong to $\Pris + \Pri$, we deduce that $\Biddersp$ is contained in $\IALosers{\IAuc} = \Bidders \setminus \IAWinners{\IAuc}$.
  Then Lemma~\ref{lem:iuap-losers-again} implies that no bidder in $\Biddersp$ is matched in any greedy MWM of $\UAP{\IAucp}$, and thus $\Prisp \subseteq \Pris + \Pri$.
\end{proof}

\comment{
  Onur: The following lemma is used only in the proof of Lemma ``IUAP threshold''.
}

\begin{lemma}
  \label{lem:uap-threshold-nondec}
  Let $\Auc = (\Bidders, \Items)$ and $\Aucp = \Auc + \Bidder$ be UAPs, and let $\Item$ be an item in $\Items$.
  Then $\AThreshold{\Auc}{\Item} \leq \AThreshold{\Aucp}{\Item}$.
\end{lemma}

\begin{proof}
  Assume for the sake of contradiction that $\AThreshold{\Auc}{\Item} > \AThreshold{\Aucp}{\Item}$.
  Then there exists a bidder $\Bidderp$ such that $\Bidderp$ does not belong to $\Bidders + \Bidder$, $\BUBid{\Bidderp} = \Set{(\Item, \Offer)}, \BPri{\Bidderp} = \Pri$, and
  \[
    \AThreshold{\Aucp}{\Item} < (\Offer, \Pri) < \AThreshold{\Auc}{\Item} .
  \]
  Since $(\Offer, \Pri) < \AThreshold{\Auc}{\Item}$, Lemma~\ref{lem:uap-threshold} implies that $\Bidderp$ is not matched in any greedy MWM of $\Auc + \Bidderp$.
  Thus Lemma~\ref{lem:uap-one-way-street} implies that $\Bidderp$ is not matched in any greedy MWM of $\Aucp + \Bidderp$.
  On the other hand, since $\AThreshold{\Aucp}{\Item} < (\Offer, \Pri)$, Lemma~\ref{lem:uap-threshold} implies that $\Bidderp$ is matched in every greedy MWM of $\Aucp + \Bidderp$, a contradiction.
\end{proof}

\comment{
  Right now, we include the assumption ``$\Pri$ does not belong to $\IAPriorities{\IAuc}$'' in the statement of Lemma~\ref{lem:iuap-threshold} below, but we don't actually use this assumption in the proof.
  The reason this assumption is present is it implies that $\Bidder$ belongs to either $\IAWinners{\IAucp}$ or $\IALosers{\IAucp}$.
  We may wish to strengthen the lemma statement to include this claim.
  Or we may wish to drop this assumption from the lemma.
  What we should do probably depends on whether we need this additional claim in any of the places where we invoke Lemma~\ref{lem:iuap-threshold}.
  I haven't checked this yet.

  Onur: I made it explicit everywhere that the bidder that is subject to this lemma belongs either to winners or to losers
}

\begin{lemma}
  \label{lem:iuap-threshold}
  Let $\IAuc = (\MBidders, \Items)$ and $\IAucp = \IAuc + \Bidder$ be IUAPs where $\Bidder = (\BidId, \Set{(\Item, \Offer)}, \Pri)$, $\Item$ is an item in $\Items$, and $\Pri$ does not belong to $\IAPriorities{\IAuc}$.
  If $\Bidder$ belongs to $\IAWinners{\IAucp}$, then $(\Offer, \Pri) > \IAThreshold{\IAuc}{\Item}$.
  If $\Bidder$ belongs to $\IALosers{\IAucp}$, then $(\Offer, \Pri) < \IAThreshold{\IAuc}{\Item}$.
\end{lemma}

\begin{proof}
  First, assume that $\Bidder$ belongs to $\IAWinners{\IAucp}$.
  Thus $\Bidder$ is matched in every greedy MWM of $\UAP{\IAucp}$, which is equal to $\UAPwItem{\IAuc}{\Item} + \Bidder$ by definition.
  Lemma~\ref{lem:uap-threshold} implies that $(\Offer, \Pri) > \AThreshold{\UAPwItem{\IAuc}{\Item}}{\Item} = \IAThreshold{\IAuc}{\Item}$, as required.

  Now assume that $\Bidder$ belongs to $\IALosers{\IAucp}$.
  Thus $\Bidder$ is not matched in any greedy MWM of $\UAP{\IAucp}$.
  Define $\Bidders$ so that $\UAP{\IAucp} = (\Bidders + \Bidder, \Items)$, and let $\Auc$ denote the UAP $(\Bidders, \Items)$.
  Lemma~\ref{lem:uap-threshold} implies that $(\Offer, \Pri) < \AThreshold{\Auc}{\Item}$.
  Lemma~\ref{lem:iuap-threshold-uap} implies that $\UAPwItem{\IAuc}{\Item} + \Bidder$ extends $\UAP{\IAucp}$, and hence that $\UAPwItem{\IAuc}{\Item}$ extends $\Auc$.
  Lemma~\ref{lem:uap-threshold-nondec} therefore implies that $$\AThreshold{\Auc}{\Item} \leq \AThreshold{\UAPwItem{\IAuc}{\Item}}{\Item} = \IAThreshold{\IAuc}{\Item}.$$
  Thus $(\Offer, \Pri) < \IAThreshold{\IAuc}{\Item}$, as required.
\end{proof}

\begin{lemma}
  \label{lem:iuap-threshold-nondec}
  Let $\IAuc = (\MBidders, \Items)$ and $\IAucp = \IAuc + \Bidder$ be IUAPs, and let $\Item$ be an item in $\Items$.
  Then $\IAThreshold{\IAuc}{\Item} \leq \IAThreshold{\IAucp}{\Item}$.
\end{lemma}

\comment{Onur: Chose $\Bidderp$ so that its priority does not belong to $\IAPriorities{\IAuc} + \BPri{\Bidder}$, and made it explicit that it belongs to either winners or losers.}

\begin{proof}
  Let $(\Offer, \Pri)$ denote $\IAThreshold{\IAuc}{\Item}$, let $(\Offerp, \Prip)$ denote $\IAThreshold{\IAucp}{\Item}$, and assume for the sake of contradiction that $(\Offer, \Pri) > (\Offerp, \Prip)$.

  Let $\Bidderp$ be a bidder $(\BidId, \Set{(\Item, \Offer)}, \Pripp)$ such that $\Pripp$ does not belong to $\IAPriorities{\IAuc} + \BPri{\Bidder}$, $\Pri > \Pripp$, and $(\Offer, \Pripp) > (\Offerp, \Prip)$.
  Let $\IAucpp$ denote $\IAuc + \Bidderp$ and let $\IAucppp$ denote $\IAucp + \Bidderp$.
  Since $\Pripp$ does not belong to $\IAPriorities{\IAuc}$, we deduce that $\Bidderp$ belongs to either $\IAWinners{\IAucpp}$ or $\IALosers{\IAucpp}$.
  Then, by Lemma~\ref{lem:iuap-threshold}, $\Bidderp$ belongs to $\IALosers{\IAucpp}$, and hence by Lemma~\ref{lem:iuap-losers-again}, $\Bidderp$ belongs to $\IALosers{\IAucppp}$.
  On the other hand, since $\Pripp$ does not belong to $\IAPriorities{\IAuc} + \BPri{\Bidder}$, Lemma~\ref{lem:iuap-priorities} implies that $\Pripp$ does not belong to $\IAPriorities{\IAucp}$, and we deduce that $\Bidderp$ belongs to either $\IAWinners{\IAucppp}$ or $\IALosers{\IAucppp}$.
  Then, Lemma~\ref{lem:iuap-threshold} implies that $\Bidderp$ belongs to $\IAWinners{\IAucppp}$, a contradiction.
\end{proof}

\begin{lemma}
  \label{lem:iuap-threshold-unchanged}
  Let $\IAuc = (\MBidders, \Items)$ and $\IAucp = \IAuc + \Bidder$ be IUAPs where $\Bidder$ belongs to $\IALosers{\IAucp}$, and let $\Item$ be an item in $\Items$.
  Then $\IAThreshold{\IAucp}{\Item} = \IAThreshold{\IAuc}{\Item}$.
\end{lemma}

\begin{proof}
  Suppose not.
  Then by Lemma~\ref{lem:iuap-threshold-nondec}, we have $\IAThreshold{\IAuc}{\Item} < \IAThreshold{\IAucp}{\Item}$.
  Let $\Pri$ denote $\BPri{\Bidder}$.
  Since $\IAucp = \IAuc + \Bidder$ and $\Bidder$ belongs to $\IALosers{\IAucp}$, we deduce that $\Pri$ does not belong to $\IAPriorities{\IAuc}$.
  Since $\Bidder$ belongs to $\IALosers{\IAucp}$, we deduce that $\Pri$ does not belong to $\IAPriorities{\IAucp}$.
  Hence Lemma~\ref{lem:iuap-priorities} implies that $\IAPriorities{\IAucp} = \IAPriorities{\IAuc}$.

  Let $\IAucpp$ denote $\IAuc + \Bidderp$ where $\Bidderp = (\BidId, \Set{(\Item, \Offerp)}, \Prip)$ is a bidder such that $\Prip$ does not belong to $\IAPriorities{\IAuc} + \Pri$ and $\IAThreshold{\IAuc}{\Item} < (\Offerp, \Prip) < \IAThreshold{\IAucp}{\Item}$.

  Let $\IAucppp$ denote $\IAucp + \Bidderp$.
  Since $\Prip$ does not belong to $\IAPriorities{\IAuc} + \Pri$, Lemma~\ref{lem:iuap-priorities} implies that $\Prip$ does not belong to $\IAPriorities{\IAucp}$, and we deduce that $\Bidderp$ belongs to either $\IAWinners{\IAucppp}$ or $\IALosers{\IAucppp}$.
  Since $(\Offerp, \Prip) < \IAThreshold{\IAucp}{\Item}$, Lemma \ref{lem:iuap-threshold} implies that $\Bidderp$ belongs to $\IALosers{\IAucppp}$.
  Hence Lemma~\ref{lem:iuap-priorities} implies that $\IAPriorities{\IAucppp} = \IAPriorities{\IAucp}$.
  Since we have established above that $\IAPriorities{\IAucp} = \IAPriorities{\IAuc}$, we deduce that $\IAPriorities{\IAucppp} = \IAPriorities{\IAuc}$.

  Since $\Prip$ does not belong to $\IAPriorities{\IAuc}$, we deduce that $\Bidderp$ belongs to either $\IAWinners{\IAucpp}$ or $\IALosers{\IAucpp}$.
  Since $(\Offerp, \Prip) > \IAThreshold{\IAuc}{\Item}$, Lemma~\ref{lem:iuap-threshold} implies that $\Bidderp$ belongs to $\IAWinners{\IAucpp}$ and hence $\Prip$ belongs to $\IAPriorities{\IAucpp}$.
  We consider two cases.

  Case 1: $\Car{\IAPriorities{\IAucpp}} \leq \Car{\IAPriorities{\IAuc}}$.
  Lemma~\ref{lem:iuap-priorities} implies that there exists a real $\Pripp$ in $\IAPriorities{\IAuc}$ that does not belong to $\IAPriorities{\IAucpp}$.
  Since $\Pri$ does not belong to $\IAPriorities{\IAuc}$, we have $\Pri \not= \Pripp$.
  Since $\IAucppp = \IAucpp + \Bidder$ and $\Pri \not= \Pripp$, Lemma~\ref{lem:iuap-priorities} implies that $\Pripp$ does not belong to $\IAPriorities{\IAucppp}$, a contradiction since $\IAPriorities{\IAucppp} = \IAPriorities{\IAuc}$.

  Case 2: $\Car{\IAPriorities{\IAucpp}} > \Car{\IAPriorities{\IAuc}}$.
  Since $\IAPriorities{\IAucppp} = \IAPriorities{\IAuc}$, we deduce that $\Car{\IAPriorities{\IAucpp}} > \Car{\IAPriorities{\IAucppp}}$.
  Since $\IAucppp = \IAucpp + \Bidder$, Lemma~\ref{lem:iuap-priorities} implies that $\Car{\IAPriorities{\IAucppp}} \geq \Car{\IAPriorities{\IAucpp}}$, a contradiction.
\end{proof}

\begin{lemma}
  \label{lem:iuap-threshold-success}
  Let $\IAuc = (\MBidders, \Items)$ and $\IAucp = \IAuc + \Bidder$ be IUAPs where $\Bidder = (\BidId, \UBid, \Pri)$ and $\Pri$ does not belong to $\IAPriorities{\IAuc}$, and let $\Item$ be an item in $\Items$.
  Assume that $(\Item, \Offer)$ belongs to $\UBid$, and that $\IAThreshold{\IAuc}{\Item} < (\Offer, \Pri)$.
  Then $\Bidder$ belongs to $\IAWinners{\IAucp}$.
\end{lemma}

\begin{proof}
  Suppose not.
  Let $\Aucp = (\Biddersp, \Items)$ denote $\UAP{\IAucp}$.
  Since $\Pri$ does not belong to $\IAPriorities{\IAuc}$, we deduce that $\Bidder$ belongs to $\Biddersp$.
  Thus $\Bidder$ belongs to $\Biddersp \setminus \IAWinners{\IAucp} = \IALosers{\IAucp}$, and so $\IAThreshold{\IAucp}{\Item} = \IAThreshold{\IAuc}{\Item}$ by Lemma~\ref{lem:iuap-threshold-unchanged}.

  Let $\IAucpp$ denote $\IAucp + \Bidderp$ where $\Bidderp = (\BidId, \Set{(\Item, \Offer)}, \Prip)$ is a bidder such that $\Prip$ does not belong to $\IAPriorities{\IAuc} + \Pri$, $\IAThreshold{\IAuc}{\Item} < (\Offer, \Prip)$, and $\Prip < \Pri$.
  Since $\Prip$ does not belong to $\IAPriorities{\IAuc} + \Pri$, we deduce that $\Bidderp$ belongs to either $\IAWinners{\IAucpp}$ or $\IALosers{\IAucpp}$.
  Then, by Lemma~\ref{lem:iuap-threshold}, $\Bidderp$ belongs to $\IAWinners{\IAucpp}$.
  Let $\Aucpp = (\Bidderspp, \Items)$ denote $\UAP{\IAucpp}$, and let $\Match$ be a greedy MWM of $\Aucpp$.
  Since $\Bidderp$ belongs to $\IAWinners{\IAucpp}$, the edge $(\Bidderp, \Item)$ belongs to $\Match$.
  Since $\Bidder$ belongs to $\IALosers{\IAucp}$, Lemma~\ref{lem:iuap-losers-again} implies that $\Bidder$ belongs to $\IALosers{\IAucpp}$, and hence that $\Bidder$ is unmatched in $\Match$.
  By Lemma~\ref{lem:uap-in-out}, we find that $(\Offer, \Pri) < (\Offer, \Prip)$ and hence $\Pri < \Prip$, a contradiction.
\end{proof}

\begin{lemma}
  \label{lem:iuap-threshold-failure}
  Let $\IAuc = (\MBidders, \Items)$ and $\IAucSub{0} = \IAuc + \Bidder$ be IUAPs where $\Bidder = (\BidId, \UBid, \Pri)$, $\Pri$ does not belong to $\IAPriorities{\IAuc}$, and $\UBid = \Set{(\ItemSub{1}, \OfferSub{1}), \dotsc, (\ItemSub{k}, \OfferSub{k})}$.
  Assume that $(\OfferSub{i}, \Pri) < \IAThreshold{\IAuc}{\ItemSub{i}}$ holds for all $i$ such that $1 \leq i \leq k$.
  Then $\Bidder$ belongs to $\IALosers{\IAucSub{0}}$.
\end{lemma}

\begin{proof}
  Suppose not.
  Since $\Pri$ does not belong to $\IAPriorities{\IAuc}$, we deduce that $\Bidder$ belongs to $\IAWinners{\IAucSub{0}}$, and hence that $\Pri$ belongs to $\IAPriorities{\IAucSub{0}}$.

  For $i$ ranging from $1$ to $k$, let $\IAucSub{i}$ denote the IUAP $\IAucSub{i-1} + \BidderSub{i}$ where $\BidderSub{i} = (\BidIdSub{i}, \Set{(\ItemSub{i}, \OfferSub{i})}, \PriSub{i})$ and $\PriSub{i}$ is a real number satisfying the following conditions:
  $\PriSub{i}$ does not belong to $\IAPriorities{\IAucSub{i-1}}$;
  $\Pri < \PriSub{i}$;
  $(\OfferSub{i}, \PriSub{i}) < \IAThreshold{\IAuc}{\ItemSub{i}}$.
  Since $\PriSub{i}$ does not belong to $\IAPriorities{\IAucSub{i-1}}$, we deduce that $\BidderSub{i}$ belongs to either $\IAWinners{\IAucSub{i}}$ or $\IALosers{\IAucSub{i}}$ for $1 \leq i \leq k$.
  Then, by Lemmas~\ref{lem:iuap-threshold} and~\ref{lem:iuap-threshold-nondec}, we deduce that $\BidderSub{i}$ belongs to $\IALosers{\IAucSub{i}}$ for $1 \leq i \leq k$.
  By repeated application of Lemma \ref{lem:iuap-priorities}, we find that $\IAPriorities{\IAucSub{i}} = \IAPriorities{\IAucSub{0}}$ for $1 \leq i \leq k$, and hence that $\Pri$ belongs to $\IAPriorities{\IAucSub{k}}$.

  We claim that $\Bidder$ belongs to $\IAWinners{\IAucSub{k}}$.
  To prove this claim, let $\MBidder$ denote the unique multibidder in $\IAucSub{k}$ for which $\MBPri{\MBidder} = \BPri{\Bidder}$.
  Let $\ell$ denote $\Car{\MBBiddersAll{\MBidder}}$, and observe that $\Bidder = \MBBidder{\MBidder}{\ell}$.
  Furthermore, since $\Pri$ does not belong to $\IAPriorities{\IAuc}$, we deduce that $\MBBidder{\MBidder}{i}$ belongs to $\IALosers{\IAuc}$ for $1 \leq i < \ell$.
  By repeated application of Lemma~\ref{lem:iuap-losers-again}, we deduce that $\MBBidder{\MBidder}{i}$ belongs to $\IALosers{\IAucSub{k}}$ for $1 \leq i< \ell$.
  Since $\Pri$ belongs to $\IAPriorities{\IAucSub{k}}$, the claim follows.

  Let $\Match$ denote a greedy MWM of $\UAP{\IAucSub{k}}$.
  Since $\Bidder$ belongs to $\IAWinners{\IAucSub{k}}$, there is a unique integer $i$, $1 \leq i \leq k$, such that $\Match$ contains edge $(\Bidder, \ItemSub{i})$.
  Let $i$ denote this integer.
  Since $\PriSub{i}$ does not belong to $\IAPriorities{\IAucSub{k}}$, we know that $\BidderSub{i}$ belongs to $\IALosers{\IAucSub{k}}$ and hence that $\BidderSub{i}$ is not matched in any greedy MWM of $\UAP{\IAucSub{k}}$.
  By Lemma~\ref{lem:uap-in-out}, we deduce that $(\OfferSub{i}, \PriSub{i}) < (\OfferSub{i}, \Pri)$.
  Hence $\PriSub{i} < \Pri$, contradicting the definition of $\PriSub{i}$.
\end{proof}

\begin{lemma}
  \label{lem:iuap-threshold-summary}
  Let $\IAucSub{0} = (\MBidders, \Items)$ be an IUAP, let $\Pri$ be a real that is not equal to the priority of any multibidder in $\MBidders$, let $k$ be a nonnegative integer, and for $i$ ranging from $1$ to $k$, let $\IAucSub{i}$ denote the IUAP $\IAucSub{i-1} + \BidderSub{i}$, where $\BPri{\BidderSub{i}} = \Pri$.
  Let $I$ denote the set of all integers $i$ in $\Set{1, \dotsc, k}$ such that there exists an item $\Item$ in $\Items$ for which $(\WPair{\BidderSub{i}}{\Item}, \Pri) > \IAThreshold{\IAucSub{0}}{\Item}$.
  If $I$ is empty, then $\Pri$ does not belong to $\IAPriorities{\IAucSub{k}}$.
  Otherwise, $\BidderSub{j}$ belongs to $\IAWinners{\IAucSub{k}}$, where $j$ denotes the minimum integer in $I$.
\end{lemma}

\begin{proof}
  If $I$ is empty, then by repeated application of Lemmas~\ref{lem:iuap-threshold-unchanged} and~\ref{lem:iuap-threshold-failure}, we find that $\BidderSub{i}$ belongs to $\IALosers{\IAucSub{i}}$ for $1 \leq i \leq k$.
  By repeated application of Lemma~\ref{lem:iuap-losers-again}, we deduce that $\BidderSub{i}$ belongs to $\IALosers{\IAucSub{k}}$ for $1 \leq i \leq k$.
  It follows that $\Pri$ does not belong to $\IAPriorities{\IAucSub{k}}$, as required.

  Now assume that $I$ is nonempty, and let $j$ denote the minimum integer in $I$.
  Arguing as in the preceding paragraph, we find that $\Pri$ does not belong to $\IAPriorities{\IAucSub{j-1}}$.
  By repeated application of Lemma \ref{lem:iuap-threshold-unchanged}, we deduce that $\IAThreshold{\IAucSub{j-1}}{\Item} = \IAThreshold{\IAucSub{0}}{\Item}$ for all items $\Item$ in $\Items$.
  Thus Lemma~\ref{lem:iuap-threshold-success} implies that $\BidderSub{j}$ belongs to $\IAWinners{\IAucSub{j}}$.
  Then, since $\BidderSub{j+1}, \dotsc, \BidderSub{k}$ all have the same priority as $\BidderSub{j}$, it is easy to argue by Lemma~\ref{lem:iuap-confluence} that $\UAP{\IAucSub{k}} = \UAP{\IAucSub{j}}$, and hence $\BidderSub{j}$ belongs to $\IAWinners{\IAucSub{k}}$, as required.
\end{proof}

\begin{proof}[Proof of Lemma~\ref{lem:iuap-threshold-george}]
  It is easy to see that the claims of the lemma follow from Lemma~\ref{lem:iuap-threshold-summary}.
\end{proof}

\comment{
  Remaining sections use Lemma~\ref{lem:iuap-exit} from Section~\ref{sec:unfold-iuap}, and Lemma~\ref{lem:iuap-threshold-george} from Section~\ref{sec:iuap-threshold}.
}

\DeclarePairedDelimiter\abs{\lvert}{\rvert}

\newcommand{\items}{{\BItemsId}}
\newcommand{\threshold}{\ThresholdId}

\newcommand{\Uap}{\UAPId}

\section{Stable Marriage with Indifferences}
\label{sec:smiw}

The \emph{stable marriage model with incomplete and weak preferences
  (SMIW)} involves a set $P$ of men and a set $Q$ of women.  The
preference relation of each man $p$ in $P$ is specified as a binary
relation $\succeq_p$ over $Q \cup \{\varnothing\}$ that satisfies
transitivity and totality, where $\varnothing$ denotes being
unmatched.
Similarly, the preference relation of each woman $q$ in $Q$ is
specified as a binary relation $\succeq_q$ over $P \cup
\{\varnothing\}$ that satisfies transitivity and totality, where
$\varnothing$ denotes being unmatched.
To allow indifferences,
the preference relations are not required to satisfy antisymmetry.
We will use $\succ_p$ and $\succ_q$ to denote the asymmetric part
of $\succeq_p$ and $\succeq_q$ respectively.

A matching is a function $\mu$ from $P$ to $Q \cup \{\varnothing\}$
such that for any woman $q$ in $Q$, there exists at most one man $p$
in $P$ for which $\mu(p) = q$. Given a matching $\mu$ and a woman $q$
in $Q$, we denote
\begin{equation*}
\mu(q) =
\begin{cases}
p & \text{if } \mu(p) = q \\
\varnothing & \text{if there is no man $p$ in $P$ such that $\mu(p) = q$}
\end{cases}
\end{equation*}

A matching $\mu$ is \emph{individually rational} if
for any man $p$ in $P$ and woman $q$ in $Q$ such that $\mu(p) = q$,
we have $q \succeq_p \varnothing$ and $p \succeq_q \varnothing$.
A pair $(p, q')$ in $P \times Q$ is said to form a
\emph{strongly blocking pair} for a matching $\mu$
if $q' \succ_p \mu(p)$ and $p \succ_{q'} \mu(q')$.
A matching is \emph{weakly stable} if it is individually rational and does not admit a strongly blocking pair.

For any matching $\mu$ and $\mu'$, we say that the binary
relation $\mu \succeq \mu'$ holds if for every man $p$ in $P$
and woman $q$ in $Q$, we have $\mu(p) \succeq_p \mu'(p)$
and $\mu(q) \succeq_q \mu'(q)$.
We let $\succ$ denote the asymmetric part of $\succeq$.
We say that a matching $\mu$ \emph{Pareto-dominates}
another matching $\mu'$ if $\mu \succ \mu'$.
We say that a matching is \emph{Pareto-optimal} if it is not Pareto-dominated by any other matching.
A matching is \emph{Pareto-stable} if it is Pareto-optimal and weakly stable.

A \emph{mechanism} is an algorithm that,
given $(P, Q, (\succeq_p)_{p \in P}, (\succeq_q)_{q \in Q})$,
produces a matching $\mu$.
A mechanism is said to be \emph{strategyproof (for the men)} if
for any man $p$ in $P$ expressing preference $\succeq_p'$ instead of
his true preference $\succeq_p$, we have $\mu(p) \succeq_p \mu'(p)$,
where $\mu$ and $\mu'$ are the matchings produced by the mechanism
given $\succeq_p$ and $\succeq_p'$, respectively, when
all other inputs are fixed.

By introducing extra men or women who prefer being unmatched
to being matched with any potential partner, we may
assume without loss of generality that the number of men
is equal to the number of women.
So, $P = \{p_1, \ldots, p_n\}$ and $Q = \{q_1, \ldots, q_n\}$.

\subsection{Algorithm}

The computation of a matching for SMIW is shown in Algorithm~\ref{alg:smiw}.
We construct an item for each woman in line~\ref{line:smiw-item},
and a multibidder for each man in line~\ref{line:smiw-bidder}
by examining the tiers of preferences of the men and
the utilities of the women.
Together with dummy items constructed in line~\ref{line:smiw-dummy},
this forms an IUAP,
from which we obtain a UAP and a greedy MWM $M_0$.
Using Lemma~\ref{lem:iuap-exit}, we argue that for any man $p_i$, exactly one of the bidders associated with $p_i$ is matched in $M_0$;
see the proof of Lemma~\ref{lem:smiw-valid}.
Finally, in line~\ref{line:smiw-return}, we use $M_0$ to determine the match of a man $p_i$ as follows, where $\Bidder$ denotes the unique bidder associated with $p_i$ that is matched in $M_0$:
if $\Bidder$ is matched in $M_0$ to the item corresponding to a woman $q_j$, then we match $p_i$ to $q_j$;
otherwise, $\Bidder$ is matched to a dummy item in $M_0$, and we leave $p_i$ unmatched.

\begin{algorithm}[htb]
\caption{}
\label{alg:smiw}
\begin{algorithmic}[1]

\State Let $p_0$ denote $\varnothing$.

\ForAll{$1 \leq j \leq n$}
  
\State Convert the preference relation $\succeq_{q_j}$ of woman $q_j$
into utility function $\psi_{q_j} \colon P \cup \{\varnothing\}
\to \Reals$ that satisfies the followings:
$\psi_{q_j}(\varnothing) = 0$;
for any $i$ and $i'$ in $\{0, 1, \ldots, n\}$, we have
$p_i \succeq_{q_j} p_{i'}$ if and only if
$\psi_{q_j}(p_i) \geq \psi_{q_j}(p_{i'})$.
This utility assignment should not depend on the preferences of the men.

\State \label{line:smiw-item}
Construct an item $v_j$ corresponding to woman $q_j$.

\EndFor

\ForAll{$n < j \leq 2 n$}

\State Let $q_j$ denote $\varnothing$.

\State \label{line:smiw-dummy}
Construct a dummy item $v_j$ corresponding to $q_j$.

\EndFor

\ForAll{$1 \leq i \leq n$}

\State Partition the set $\{1, \ldots, n\} \cup \{n+i\}$ of woman indices
into tiers $\tau_i(1), \ldots, \tau_i(K_i)$ according to the
preference relation of man $p_i$, such that for any $j$ in $\tau_i(k)$
and $j'$ in $\tau_i(k')$, we have $q_j \succeq_{p_i} q_{j'}$ if and
only if $k \leq k'$.

\State For $j$ in $\{1, \ldots, n\} \cup \{n+i\}$,
denote tier number $\kappa_i(q_j)$
as the unique $k$ such that $j$ in $\tau_i(k)$.

\State \label{line:smiw-bidder}
Construct a multibidder $t_i = (\sigma_i , z_i)$
with priority $z_i = i$ corresponding to man $p_i$.
The multibidder $t_i$ has $K_i$ bidders.
For each bidder $\sigma_i(k)$ we define
$\items(\sigma_i(k))$ as $\{ v_j \mid j \in \tau_i(k) \}$
and $w(\sigma_i(k), q_j)$ as $\psi_{q_j}(p_i)$,
where $\psi_{q_{n+i}}(p_i)$ is defined to be $0$.

\EndFor

\State $(T, V) = (\{ t_i \mid 1 \leq i \leq n \},
\{v_j \mid 1 \leq j \leq 2 n\} )$.

\State \label{line:smiw-uap}
$(U, V) = \Uap(T, V)$.

\State \label{line:smiw-greedy}
Compute a greedy MWM $M_0$ of UAP $(U, V)$
  as described in Section~\ref{sec:inc-hungarian-step}.

\State \label{line:smiw-return}
Output matching $\mu$ such that for all $1 \leq i \leq n$
and $1 \leq j \leq 2 n$, we have $\mu(p_i) = q_j$ if and only if
$\sigma_i(k)$ is matched to item $v_j$ in $M_0$ for some $k$.
\end{algorithmic}
\end{algorithm}

In Lemma~\ref{lem:smiw-rational}, we prove individually rationality
by arguing that the dummy items ensure that no man
or woman is matched to an unacceptable partner.
In Lemma~\ref{lem:smiw-stable}, we prove weak stability
using the properties of a greedy MWM.
In Lemmas~\ref{lem:smiw-pareto} and~\ref{lem:smiw-dominate},
we prove Pareto-optimality by showing that any matching that
Pareto-dominates the output matching induces another MWM
that contradicts the greediness of the MWM produced by the algorithm.
In Lemma~\ref{lem:smiw-threshold}, we establish two properties
of IUAP thresholds that are used to show strategyproofness
in Theorem~\ref{thm:smiw}.

\begin{lemma} \label{lem:smiw-valid}
Algorithm~\ref{alg:smiw} produces a valid matching.
\begin{proof}
First, we show that for any man $p_i$ where $1 \leq i \leq n$,
there exists at most one $j$ in $\{1, \ldots, 2 n\}$
such that bidder $\sigma_i(k)$ is matched to
item $v_j$ in $M_0$ for some $k$.
For the sake of contradiction, suppose bidder $\sigma_i(k)$
is matched to item $v_j$ and bidder $\sigma_i(k')$ is matched
to item $v_{j'}$ in $M_0$ for some $k$ and $k'$ where $j \neq j'$.
By Lemma~\ref{lem:iuap-exit}, we have $k \leq k'$ and $k' \leq k$.
Therefore, bidder $\sigma_i(k) = \sigma_i(k')$ is matched in $M_0$
to both $v_j$ and $v_{j'}$, which is a contradiction.

Next, we show that for any man $p_i$ where $1 \leq i \leq n$,
there exists at least one $j$ in $\{1, \ldots, 2 n\}$
such that bidder $\sigma_i(k)$ is matched to
item $v_j$ in $M_0$ for some $k$. For the sake of contradiction,
suppose bidder $\sigma_i(k)$ is unmatched in $M_0$ for all $k$.
Let $j$ denote $n+i$ and let $k$ denote $\kappa_i(q_j)$.
By Lemma~\ref{lem:iuap-exit},
the set $\Bidders$ contains bidder $\sigma_i(k)$.
Since both bidder $\sigma_i(k)$ and item $v_j$ are unmatched by $M_0$,
adding the pair $(\sigma_i(k), v_j)$ to $M_0$ gives
a matching of $(U, V)$ with the same weight and larger cardinality.
This contradicts the fact that $M_0$ is a greedy MWM of $(U, V)$.

This shows that $\mu(p_i)$ is well-defined for all men $p_i$
where $1 \leq i \leq n$.
Furthermore, since each item $v_j$ where $1 \leq j \leq n$ is matched to
at most one bidder in $M_0$, each woman $q_j$ is matched to
at most one man $p_i$ in $\mu$ where $1 \leq i \leq n$.
Hence, $\mu$ is a valid matching.
\end{proof}
\end{lemma}

\begin{lemma} \label{lem:smiw-rational}
Algorithm~\ref{alg:smiw} produces an individually rational matching.
\begin{proof}
We have shown in Lemma~\ref{lem:smiw-valid} that $\mu$ is a valid matching.
Consider man $p_i$ and woman $q_j$ such that $\mu(p_i) = q_j$,
where $i$ and $j$ belong to $\{1, \ldots, n\}$.
Let $k$ denote $\kappa_i(q_j)$ and let $k'$ denote $\kappa_i(q_{n+i})$.
It suffices to show that $k \leq k'$ and $\psi_{q_j}(p_i) \geq 0$.

Since $\mu(p_i) = q_j$,
bidder $\sigma_i(k)$ is matched to item $v_j$ in $M_0$.
Since $M_0$ is an MWM,
we have $\psi_{q_j}(p_i) = w(\sigma_i(k), v_j) \geq 0$.

It remains to show that $k \leq k'$.
For the sake of contradiction, suppose $k > k'$.
Since bidder $\sigma_i(k)$ is matched to item $v_j$ in $M_0$,
by Lemma~\ref{lem:iuap-exit} the set $\Bidders$ contains bidder $\sigma_i(k')$.
Since bidder $\sigma_i(k')$ is not matched in $M_0$, 
the dummy item $v_{n+i}$ is also not matched in $M_0$.
Hence, adding the pair $(\sigma_i(k'), v_{n+i})$ to $M_0$
gives a matching in $(U, V)$ with the same weight
and larger cardinality.
This contradicts the fact that $M_0$ is a greedy MWM of $(U, V)$.
\end{proof}
\end{lemma}

\begin{lemma} \label{lem:smiw-stable}
Algorithm~\ref{alg:smiw} produces a weakly stable matching.
\begin{proof}
By Lemma~\ref{lem:smiw-rational}, it remains only to show that $\mu$ does not admit a strongly blocking pair.
Consider man $p_i$ and woman $q_{j'}$, where $i$ and
$j'$ belong to $\{1, \ldots, n\}$.
We want to show that $(p_i, q_{j'})$
does not form a strongly blocking pair.
Let $q_j$ denote $\mu(p_i)$ and let $p_{i'}$ denote $\mu(q_{j'})$,
where $j$ belongs to $\{1, \ldots, n\} \cup \{n+i\}$
and $i'$ belongs to $\{0, 1, \ldots, n\}$.
It suffices to show that either
$\kappa_i(q_j) \leq \kappa_i(q_{j'})$ or
$\psi_{q_{j'}}(p_{i'}) \geq \psi_{q_{j'}}(p_i)$.
For the sake of contradiction, suppose 
$\kappa_i(q_j) > \kappa_i(q_{j'})$ and
$\psi_{q_{j'}}(p_{i'}) < \psi_{q_{j'}}(p_i)$.
Let $k$ denote $\kappa_i(q_j)$ and let $k'$ denote $\kappa_i(q_{j'})$.
Since $\sigma_i(k)$ is matched in $M_0$ and $k' < k$,
Lemma~\ref{lem:iuap-exit} implies that the set $\Bidders$ contains bidder $\sigma_i(k')$ and that $\sigma_i(k')$ is unmatched in $M_0$.
We consider two cases.

Case 1: $i' = 0$. Then $\psi_{q_{j'}}(p_i) > \psi_{q_{j'}}(p_{i'}) = 0$.
Since neither bidder $\sigma_i(k')$ nor item $v_{j'}$ is matched in $M_0$,
adding the pair $(\sigma_i(k'), v_{j'})$ to $M_0$ gives a matching of $(U, V)$
with a larger weight. This contradicts the fact that $M_0$ is an MWM
of $(U, V)$.

Case 2: $i' \neq 0$.
Since $p_{i'} = \mu(q_{j'})$, there exists $k''$ 
such that bidder $\sigma_{i'}(k'')$ is matched to $v_{j'}$ in $M_0$.
Since $\sigma_i(k')$ is unmatched in $M_0$, the matching $M_0 - (\sigma_{i'}(k''), v_{j'}) + (\sigma_i(k'), v_{j'})$ is a matching of $(U, V)$ with weight $\WMatch{M_0} - \psi_{q_{j'}}(p_{i'}) + \psi_{q_{j'}}(p_i)$, which is greater than $\WMatch{M_0}$.
This contradicts the fact that $M_0$ is an MWM
of $(U, V)$.
\end{proof}
\end{lemma}

\begin{lemma} \label{lem:smiw-pareto}
Let $\mu$ be the matching produced by Algorithm~\ref{alg:smiw}
and let $\mu'$ be a matching such that
$\mu'(p) \succeq_p \mu(p)$ for every man $p$ in $P$ and
\begin{equation*}
\sum_{q \in Q} \psi_q(\mu'(q)) \geq \sum_{q \in Q} \psi_q(\mu(q)).
\end{equation*}
Then $\mu(p) \succeq_p \mu'(p)$ for every man $p$ in $P$ and
\begin{equation*}
\sum_{q \in Q} \psi_q(\mu'(q)) = \sum_{q \in Q} \psi_q(\mu(q)).
\end{equation*}
\begin{proof}
For any $i$ such that $1 \leq i \leq n$, let $k_i$ denote
$\kappa_i(\mu(p_i))$ and let $k_i'$ denote $\kappa_i(\mu'(p_i))$.

Below we use $\mu'$ to construct an MWM $M_0'$ of $(U, V)$.
We give the construction of $M_0'$ first, and then argue that $M_0'$ is an MWM of $(U, V)$.
Let $M_0'$ denote the set of bidder-item pairs $(\sigma_{i}(k_i'), v_j)$
such that $\mu'(p_i) = q_j$ where $i$ in $\{1, \ldots, n\}$
and $j$ in $\{1, \ldots, n\} \cup \{n+i\}$.
It is easy to see that $M_0'$ is a valid matching.
Notice that for any $1 \leq i \leq n$,
since $\mu'(p_i) \succeq_{p_i} \mu(p_i)$, we have $k'_i \leq k_i$.
So, by Lemma~\ref{lem:iuap-exit}, the set $\Bidders$ contains all bidders
$\sigma_{i}(k_i')$. Hence, $M_0'$ is a matching of $(U, V)$.
Furthermore, it is easy to see that
\begin{equation*}
w(M_0') = \sum_{1 \leq j \leq n} \psi_{q_j}(\mu'(q_j))
\geq \sum_{1 \leq j \leq n} \psi_{q_j}(\mu(q_j)) = w(M_0).
\end{equation*}
Thus $M_0'$ is an MWM of $(U, V)$, and we have
\begin{equation*}
\sum_{1 \leq j \leq n} \psi_{q_j}(\mu'(q_j))
= \sum_{1 \leq j \leq n} \psi_{q_j}(\mu(q_j)).
\end{equation*}
Furthermore, $M_0'$ is an MCMWM of $(U, V)$
because both $M_0'$ and $M_0$ have cardinality equal to $n$.
Also, $M_0'$ is a greedy MWM of $(U, V)$,
because both $M_0'$ and $M_0$ have priorities equal to
$\sum_{1 \leq i \leq n} z_i$.
Hence, for each $1 \leq i \leq n$, we have $k_i \leq k_i'$
by Lemma~\ref{lem:iuap-exit}. Thus, 
$\mu(p_i) \succeq_{p_i} \mu'(p_i)$ for all $1 \leq i \leq n$.
\end{proof}
\end{lemma}

\begin{lemma} \label{lem:smiw-dominate}
Let $\mu$ be the matching produced by Algorithm~\ref{alg:smiw}
and $\mu'$ be a matching such that $\mu' \succeq \mu$.
Then, $\mu \succeq \mu'$.
\begin{proof}
Since $\mu' \succeq \mu$, we have $\mu'(p_i) \succeq_{p_i} \mu(p_i)$
and $\psi_{q_j}(\mu'(q_j)) \geq \psi_{q_j}(\mu(q_j))$
for every $i$ and $j$ in $\{1, \ldots, n\}$.
So, by Lemma~\ref{lem:smiw-pareto}, we have
$\mu(p_i) \succeq_{p_i} \mu'(p_i)$ for every $i$ in $\{1, \ldots, n\}$ and
\begin{equation*}
\sum_{1 \leq j \leq n} \psi_{q_j}(\mu'(q_j))
= \sum_{1 \leq j \leq n} \psi_{q_j}(\mu(q_j)).
\end{equation*}
Therefore, $\psi_{q_j}(\mu'(q_j)) = \psi_{q_j}(\mu(q_j))$
for every $j$ in $\{1, \ldots, n\}$.
This shows that $\mu \succeq \mu'$.
\end{proof}
\end{lemma}

\begin{lemma} \label{lem:smiw-threshold}
Consider Algorithm~\ref{alg:smiw}.
Suppose $\mu(p_i) = q_j$, where $1 \leq i \leq n$ and
$j$ belongs to $\{1, \ldots, n\} \cup \{n + i\}$.
Then, we have
\begin{equation} \label{eq:smiw-threshold-success}
(\psi_{q_j}(p_i), i) \geq \threshold((T-t_i, V), v_j).
\end{equation}
Furthermore, for all $j'$ in $\{1, \ldots, n\} \cup \{n + i\}$
such that $\kappa_i(q_{j'}) < \kappa_i(q_j)$, we have
\begin{equation} \label{eq:smiw-threshold-failure}
(\psi_{q_{j'}}(p_i), i) < \threshold((T-t_i, V), v_{j'}).
\end{equation}
\begin{proof}
Let $k$ denote $\kappa_i(q_j)$.
Since $\mu(p_i) = q_j$,
we know that bidder $\sigma_i(k)$ is matched to item $v_j$ in $M_0$.
So, inequality~(\ref{eq:iuap-threshold-success}) 
of Lemma~\ref{lem:iuap-threshold-george} implies
inequality~(\ref{eq:smiw-threshold-success}), because
$w(\sigma_i(k), v_j) = \psi_{q_j}(p_i)$
and $z_i = i$.

Now, suppose $\kappa_i(q_{j'}) < \kappa_i(q_j)$.
Let $k'$ denote $\kappa_i(q_{j'})$.
Since $k' < k$, inequality~(\ref{eq:iuap-threshold-failure})
of Lemma~\ref{lem:iuap-threshold-george} implies
inequality~(\ref{eq:smiw-threshold-failure}), because
$w(\sigma_i(k'), v_{j'}) = \psi_{q_{j'}}(p_i)$
and $z_i = i$.
\end{proof}
\end{lemma}

\begin{theorem} \label{thm:smiw}
Algorithm~\ref{alg:smiw} is a strategyproof Pareto-stable mechanism for
the stable marriage problem with incomplete and weak preferences
(for any fixed choice of utility assignment).
\begin{proof}
We have shown in Lemma~\ref{lem:smiw-stable} that the algorithm
produces a weakly stable matching. Moreover, Lemma~\ref{lem:smiw-dominate}
shows that the weakly stable matching produced is not Pareto-dominated
by any other matching. Hence, the algorithm produces a Pareto-stable matching.
It remains to show that the algorithm is a strategyproof mechanism.

Suppose man $p_i$ expresses $\succeq_{p_i}'$ instead of his true
preference relation $\succeq_{p_i}$, where $1 \leq i \leq n$.  Let
$\mu$ and $\mu'$ be the resulting matchings given $\succeq_{p_i}$ and
$\succeq_{p_i}'$, respectively.  Let $q_j$ denote $\mu(p_i)$ and let
$q_{j'}$ denote $\mu'(p_i)$, where $j$ and $j'$ belong to $\{1,
\ldots, n\} \cup \{n+i\}$.  Let $k$ denote $\kappa_i(q_j)$ and let
$k'$ denote $\kappa_i(q_{j'})$, where $\kappa_i(\cdot)$ denotes the
tier number with respect to $\succeq_{p_i}$.  It suffices to show that
$k \leq k'$.  For the sake of contradiction, suppose $k > k'$.

Let $(T, V)$ be the IUAP,
let $t_i$ be the multibidder corresponding to man $p_i$, and let
$v_{j'}$ be the item corresponding to woman $q_{j'}$
constructed in the algorithm given input $\succeq_{p_i}$.
Since $\mu(p_i) = q_j$,
by inequality~(\ref{eq:smiw-threshold-failure}) of Lemma~\ref{lem:smiw-threshold},
we have
\begin{equation*}
(\psi_{q_{j'}}(p_i), i) < \threshold((T - t_i, V), v_{j'}).
\end{equation*}

Now, consider the behavior of the algorithm when preference relation
$\succeq_{p_i}$ is replaced with $\succeq_{p_i}'$.  Let $(T',
V')$ be the IUAP, let $t'_i$ be the multibidder corresponding to man
$p_i$, and let $v'_{j'}$ be the item corresponding to woman $q_{j'}$
constructed in the algorithm given input $\succeq_{p_i}'$.  Since
$\mu'(p_i) = q_{j'}$, by inequality~(\ref{eq:smiw-threshold-success}) of
Lemma~\ref{lem:smiw-threshold}, we have
\begin{equation*}
(\psi_{q_{j'}}(p_i), i) \geq \threshold((T' - t'_i, V'), v_{j'}').
\end{equation*}

Notice that in Algorithm~\ref{alg:smiw}, the only part of the IUAP
instance that depends on the preferences of man $p_i$ is the
multibidder corresponding to man $p_i$. In particular, we have
$T - t_i = T' - t'_i$, $V = V'$, and $v_{j'} = v_{j'}'$.
Hence, we get
\begin{align*}
(\psi_{q_{j'}}(p_i), i) <{} & \threshold((T - t_i, V), v_{j'}) \\
={} & \threshold((T' - t'_i, V'), v_{j'}') \\
\leq{} & (\psi_{q_{j'}}(p_i), i),
\end{align*}
which is a contradiction.
\end{proof}
\end{theorem}

\section{College Admissions with Indifferences}
\label{sec:caw}

The \emph{college admissions model with weak preferences (CAW)}
involves a set $P$ of students and a set $Q$ of colleges.
The preference relation of each student $p$ in $P$
is specified as a binary relation $\succeq_p$ over
$Q \cup \{\varnothing\}$
that satisfies transitivity and totality,
where $\varnothing$ denotes being unmatched.
The preference relation of each college $q$ in $Q$
over individual students is specified as
a binary relation $\succeq_q$ over $P \cup \{\varnothing\}$
that satisfies transitivity and totality,
where $\varnothing$ denotes being unmatched.
Each college $q$ in $Q$ has an associated integer capacity $c_q>0$.
We will use $\succ_p$ and $\succ_q$ to denote the asymmetric parts
of $\succeq_p$ and $\succeq_q$, respectively.

The colleges' preference relation over individual students
can be extended to group preference using responsiveness.
We say that a transitive and reflexive relation
$\succeq_q'$ over the power set $2^P$ is \emph{responsive
to the preference relation $\succeq_q$} if the following conditions
hold:
for any $S \subseteq P$ and $p$ in $P \setminus S$,
we have $p \succeq_q \varnothing$ if and only if $S \cup \{p\} \succeq_q' S$;
for any $S \subseteq P$ and any $p$ and $p'$ in $P \setminus S$,
we have $p \succeq_q p'$ if and only if $S \cup \{p\} \succeq_q' S \cup \{p'\}$.
Furthermore, we say that a relation $\succeq_q'$ is
\emph{minimally responsive to
the preference relation $\succeq_q$}
if it is responsive to the preference relation $\succeq_q$
and does not strictly contain another relation that is responsive
to the preference relation $\succeq_q$.

A \emph{(capacitated) matching} is a function $\mu$ from $P$ to
$Q \cup \{ \varnothing \}$ such that for any college $q$ in $Q$,
there exists at most $c_q$ students $p$ in $P$ for which $\mu(p) = q$.
Given a matching $\mu$ and a college $q$ in $Q$, we let
$\mu(q)$ denote $\{p \in P \mid \mu(p) = q\}$.

A matching $\mu$ is \emph{individually rational} if
for any student $p$ in $P$ and college $q$ in $Q$ such that $\mu(p) = q$,
we have $q \succeq_p \varnothing$ and $p \succeq_q \varnothing$.
A pair $(p', q)$ in $P \times Q$ is said to form a
\emph{strongly blocking pair} for a matching $\mu$
if $q \succ_{p'} \mu(p')$ and at least one of the following two conditions holds:
(1) there exists a student $p$ in $P$ such that $\mu(p) = q$ and $p' \succ_q p$;
(2) $\abs{\mu(q)} < c_q$ and $p' \succ_q \varnothing$.
A matching is \emph{weakly stable} if it is individually rational and does not admit a strongly blocking pair.

Let $\succeq_q'$ be the group preference associated with college $q$ in $Q$.
For any matching $\mu$ and $\mu'$,
we say that the binary relation $\mu \succeq \mu'$ holds
if for every student $p$ in $P$ and college $q$ in $Q$, we have
$\mu(p) \succeq_p \mu'(p)$ and
$\mu(q) \succeq_q' \mu'(q)$.
We let $\succ$ denote the asymmetric part of $\succeq$.
We say that a matching $\mu$ \emph{Pareto-dominates}
another matching $\mu'$ if $\mu \succ \mu'$.
We say that a matching is \emph{Pareto-optimal} if it is not Pareto-dominated by any other matching.
A matching is \emph{Pareto-stable} if it is Pareto-optimal and weakly stable.

A \emph{mechanism} is an algorithm that, given $(P, Q,
(\succeq_p)_{p \in P}, (\succeq_q)_{q \in Q}, (c_q)_{q \in Q})$,
produces a matching $\mu$.
A mechanism is said to be \emph{strategyproof (for the students)} if
for any student $p$ in $P$ expressing preference $\succeq_p'$ instead of
their true preference $\succeq_p$, we have $\mu(p) \succeq_p \mu'(p)$,
where $\mu$ and $\mu'$ are the matchings produced by the mechanism
given $\succeq_p$ and $\succeq_p'$, respectively, when
all other inputs are fixed.

Without loss of generality, we may assume that the number of students
equals the total capacity of the colleges.
So, $P = \{p_i\}_{1 \leq i \leq \abs{P}}$ and
$Q = \{q_j\}_{1 \leq j \leq \abs{Q}}$ such that
$\abs{P} = \sum_{1 \leq j \leq \abs{Q}} c_{q_j}$.

\subsection{Algorithm}

The computation of a matching for CAW is shown in Algorithm~\ref{alg:caw}.
We transform each student to a man in line~\ref{line:caw-man},
and each slot of a college to a woman in line~\ref{line:caw-woman}.
This forms an SMIW.
Using this SMIW, we produce a matching by invoking
Algorithm~\ref{alg:smiw} in lines~\ref{line:caw-smiw} and~\ref{line:caw-return}.

\begin{algorithm}[htb]
\caption{}
\label{alg:caw}
\begin{algorithmic}[1]
\State \label{line:caw-man}
For each $1 \leq i \leq \abs{P}$,
construct man $p'_i$ corresponding to student $p_i$.

\State \label{line:caw-woman}
For each $1 \leq j \leq \abs{Q}$,
construct women $q'_{j 1}, \ldots, q'_{j c}$ corresponding to college $q_j$
with capacity $c = c_{q_j}$.

\State $(P', Q') = (\{p'_i \mid 1 \leq i \leq \abs{P}\},
\{ q'_{j k} \mid 1 \leq j \leq \abs{Q} \text{ and } 1 \leq k \leq c_{q_j} \})$.

\State Let $p_0$ denote $\varnothing$. Let $p'_0$ denote $\varnothing$.

\State Let $q_0$ denote $\varnothing$. Let $q'_{0 0}$ denote $\varnothing$.

\State For each $1 \leq i \leq \abs{P}$, define the preference relation
$\succeq_{p'_i}$ over $Q' \cup \{q'_{0 0}\}$ for man $p'_i$ using the
preference relation of his corresponding student, such that $q'_{j k}
\succeq_{p'_i} q'_{j' k'}$ if and only if $q_j \succeq_{p_i} q_{j'}$.

\State For each $1 \leq j \leq \abs{Q}$ and $1 \leq k \leq c_{q_j}$,
define the preference relation $\succeq_{q'_{j k}}$ over $P' \cup \{p'_0\}$
for woman $q'_{j k}$ using the preference relation of her corresponding college,
such that
$p'_i \succeq _{q'_{j k}} p'_{i'}$ if and only if $p_i \succeq_{q_j} p_{i'}$.

\State \label{line:caw-smiw}
Compute matching $\mu_0$ for SMIW $(P', Q',
(\succeq_{p'})_{p' \in P'}, (\succeq_{q'})_{q' \in Q'})$
using Algorithm~\ref{alg:smiw},
where we require the utility functions associated with the same
college to be the same.

\State \label{line:caw-return}
Output matching $\mu$, such that for all $1 \leq i \leq \abs{P}$
and $0 \leq j \leq \abs{Q}$,
we have $\mu(p_i) = q_j$ if and only if $\mu_0(p'_i) = q'_{j k}$ for some $k$.
\end{algorithmic}
\end{algorithm}

\begin{lemma} \label{lem:caw-ir}
Algorithm~\ref{alg:caw} produces an individually rational matching.

\begin{proof}
It is easy to see that $\mu$ satisfies the capacity constraints
because each college $q_j$ is associated with $c_{q_j}$ women $q'_{j k}$
and each woman can be matched with at most one man in $\mu_0$
by Lemma~\ref{lem:smiw-valid}.

The individual rationality of $\mu$ follows from the individual
rationality of $\mu_0$.
Let $p_i$ in $P$ and $q_j$ in $Q$ such that $\mu(p_i) = q_j$.
Then $\mu_0(p'_i) = q'_{j k}$ for some $k$.
By Lemma~\ref{lem:smiw-rational}, we have
$q'_{j k} \succeq_{p'_i} \varnothing$ and
$p'_i \succeq_{q'_{j k}} \varnothing$.
Hence, $q_j \succeq_{p_i} \varnothing$ and
$p_i \succeq_{q_j} \varnothing$.
\end{proof}
\end{lemma}

\begin{lemma} \label{lem:caw-stable}
Algorithm~\ref{alg:caw} produces a weakly stable matching.

\begin{proof}
By Lemma~\ref{lem:caw-ir}, it remains only to show that $\mu$ does not admit a strongly blocking pair.
Consider student $p_{i'}$ in $P$ and college $q_j$ in $Q$.
In what follows, we use the weak stability of $\mu_0$ to show that
$(p_{i'}, q_j)$ does not form a strongly blocking pair.

Let $q'_{j' k'}$ denote $\mu_0(p'_{i'})$.
It is possible that $q'_{j' k'} = \varnothing$, in which case $j' = k' = 0$.
For $1 \leq k \leq c_{q_j}$, let $p'_{i_k}$ denote $\mu_0(q'_{j k})$,
where $p'_{i_k}$ belongs to $P' \cup \{ p'_0 \}$.
By Lemma~\ref{lem:smiw-stable}, for any $1 \leq k \leq c_{q_j}$,
either $q'_{j' k'} \succeq_{p'_{i'}} q'_{j k}$
or $p'_{i_k} \succeq_{q'_{j k}} p'_{i'}$,
for otherwise $(p'_{i'}, q'_{j k})$ forms a strongly blocking pair.

Suppose $q'_{j' k'} \succeq_{p'_{i'}} q'_{j k}$ for some $1 \leq k \leq c_{q_j}$.
Then $q_{j'} \succeq_{p_{i'}} q_j$,
and hence $(p_{i'}, q_j)$ does not form a strongly blocking pair.

Otherwise, $p'_{i_k} \succeq_{q'_{j k}} p'_{i'}$ for all $1 \leq k \leq c_{q_j}$.
Then $p_{i_k} \succeq_{q_j} p_{i'}$ for all $1 \leq k \leq c_{q_j}$.
In particular, we have $p_{i_k} \succeq_{q_j} p_{i'}$
for all students $p_{i_k}$ in $P$ such that $\mu(p_{i_k}) = q_j$.
Furthermore, if $\abs{\mu(q_j)} < c_{q_j}$,
then $p_{i_k} = \varnothing$ for some $1 \leq k \leq c_{q_j}$.
Hence $\varnothing \succeq_{q_j} p_{i'}$.
It follows that $(p_{i'}, q_j)$ does not form a strongly blocking pair.
\end{proof}
\end{lemma}

\begin{lemma} \label{lem:caw-dominate}
Suppose that for every college $q$ in $Q$, the group preference
relation $\succeq_q'$ is minimally responsive to $\succeq_q$.
Let $\mu$ be the matching produced by Algorithm~\ref{alg:caw}
and let $\mu'$ be a matching such that $\mu' \succeq \mu$.
Then $\mu \succeq \mu'$.

\begin{proof}
Since $\mu'$ is a matching that satisfies the capacity constraints,
we can construct an SMIW matching
$\mu_0' \colon P' \to Q' \cup \{q'_{0 0}\}$
such that for all $1 \leq i \leq \abs{P}$ and $0 \leq j \leq \abs{Q}$,
we have $\mu'(p_i) = q_j$ if and only if $\mu_0(p'_i) = q'_{j k}$
for some $k$.

Since $\mu' \succeq \mu$, we have $\mu'(p_i) \succeq_{p_i} \mu(p_i)$
for every $1 \leq i \leq \abs{P}$
and $\mu'(q_j) \succeq_{q_j}' \mu(q_j)$
for every $1 \leq j \leq \abs{Q}$.
Thus $\mu_0'(p'_i) \succeq_{p'_i} \mu_0(p'_i)$
for every $1 \leq i \leq \abs{P}$ and
\begin{equation*}
\sum_{1 \leq k \leq c_{q_j}} \psi_{q'_{j k}}(\mu_0'(q'_{j k}))
\geq \sum_{1 \leq k \leq c_{q_j}} \psi_{q'_{j k}}(\mu_0(q'_{j k}))
\end{equation*}
for every $1 \leq j \leq \abs{Q}$.
Hence, by Lemma~\ref{lem:smiw-pareto}, we have
$\mu_0(p'_i) \succeq_{p'_i} \mu_0'(p'_i)$
for every $1 \leq i \leq \abs{P}$ and 
\begin{equation*}
\sum_{1 \leq j \leq \abs{Q}} \sum_{1 \leq k \leq c_{q_j}}
\psi_{q'_{j k}}(\mu_0'(q'_{j k}))
= \sum_{1 \leq j \leq \abs{Q}} \sum_{1 \leq k \leq c_{q_j}}
\psi_{q'_{j k}}(\mu_0(q'_{j k})).
\end{equation*}
Therefore, we have $\mu(p_i) \succeq_{p_i} \mu'(p_i)$
for every $1 \leq i \leq \abs{P}$ and
\begin{equation*}
\sum_{1 \leq k \leq c_{q_j}} \psi_{q'_{j k}}(\mu_0'(q'_{j k}))
= \sum_{1 \leq k \leq c_{q_j}} \psi_{q'_{j k}}(\mu_0(q'_{j k}))
\end{equation*}
for every $1 \leq j \leq \abs{Q}$.
We conclude that $\mu(q_j) \succeq_{q_j}' \mu'(q_j)$
for every $1 \leq j \leq \abs{Q}$.
Thus $\mu \succeq \mu'$.
\end{proof}
\end{lemma}

\begin{theorem} \label{thm:caw} Suppose that for every college $q$ in
  $Q$, the group preference relation $\succeq_q'$ is minimally
  responsive to $\succeq_q$.  Algorithm~\ref{alg:caw} is a
  strategyproof Pareto-stable mechanism for the college admissions
  problem with weak preferences (for any fixed choice of utility
  assignment).

\begin{proof}
  We have shown in Lemma~\ref{lem:caw-stable} that Algorithm~\ref{alg:caw}
  produces a weakly stable matching. Moreover,
  Lemma~\ref{lem:caw-dominate} shows that the weakly stable matching
  produced is not Pareto-dominated by any other matching. Hence,
  Algorithm~\ref{alg:caw} produces a Pareto-stable matching.

  To show that Algorithm~\ref{alg:caw} provides a strategyproof mechanism,
  suppose student $p_i$ expresses $\succeq_{p_i}'$ instead of their
  true preference relation $\succeq_{p_i}$, where $1 \leq i \leq
  \abs{P}$.  Let $\mu$ and $\mu'$ be the matchings produced by Algorithm~\ref{alg:caw} given
  $\succeq_{p_i}$ and $\succeq_{p_i}'$, respectively.  Let $\mu_0$ and
  $\mu_0'$ be the SMIW matching produced
  by the call to Algorithm~\ref{alg:smiw} (line \ref{line:caw-smiw} of Algorithm~\ref{alg:caw}) given
  $\succeq_{p_i}$ and $\succeq_{p_i}'$, respectively.

Notice that in Algorithm~\ref{alg:caw}, the only part
of the stable marriage instance that depends on the preferences of
student $p_i$ is the preference relation corresponding to man $p'_i$.
Since Algorithm~\ref{alg:smiw} is strategyproof by 
Theorem~\ref{thm:smiw}, we have $\mu_0(p'_i) \succeq_{p'_i} \mu_0'(p'_i)$
where $\succeq_{p'_i}$ is the preference relation of man $p'_i$
in the algorithm given $\succeq_{p_i}$.
Hence, $\mu(p_i) \succeq_{p_i} \mu'(p_i)$.
\end{proof}
\end{theorem}

We remark that our algorithm admits an $O(n^4)$-time
implementation, where $n$ is the sum of the number of students and the
total capacities of all the colleges, because the reduction from CAW
to IUAP takes $O(n^2)$ time, and 
lines~\ref{line:smiw-uap} and~\ref{line:smiw-greedy}
of Algorithm~\ref{alg:smiw}
can be implemented in $O(n^4)$ time using the version of the
incremental Hungarian method discussed in Sections~\ref{sec:inc-hungarian-step} and~\ref{sec:unfold-iuap}.

\subsection{Further Discussion}
\label{sec:caw-further-disc}

In our SMIW and CAW algorithms, we transform the preference relations
of the women and colleges into real-valued utility functions.  One way
to do this is to take
\begin{equation*}
\psi_q(p) = \abs{ \{p' \in P \cup \{\varnothing\}  \colon p \succeq_q p'\} }
- \abs{ \{p' \in P \cup \{\varnothing\}  \colon \varnothing \succeq_q p'\} }.
\end{equation*}
This is by no means the
only way. In fact, different ways of assigning the utilities can
affect the outcome. Nonetheless, our mechanisms remain strategyproof
for the men as long as the utility assignment is fixed and independent
of the preferences of the men, as shown in Theorems~\ref{thm:smiw}
and~\ref{thm:caw}.

We can also consider the scenario where each college expresses their
preferences directly in terms of a utility function instead of a
preference relation.  Such utility functions provide another way
to extend preferences over individuals to group preferences.  If a
college $q$ expresses the utility function $\psi_q$ over individual
students in $P \cup \{ \varnothing \}$, we can define the \emph{group
  preference induced by additive utility $\psi_q$} as a binary
relation $\succeq_q'$ over $2^P$ such that
$S \succeq_q' S'$ if and only if
\begin{equation*}
\sum_{p \in S} \psi_q(p)
\geq \sum_{p \in S'} \psi_q(p).
\end{equation*}
Our algorithm can accept such utility functions as input in lieu
of constructing them by some utility assignment method.  It is not
hard to see that the mechanism remains Pareto-stable and strategyproof
when the group preferences of the colleges are induced by additive
utilities.

% Bibliography
\bibliography{refs}

\begin{thebibliography}{16}
\providecommand{\natexlab}[1]{#1}
\providecommand{\url}[1]{\texttt{#1}}
\expandafter\ifx\csname urlstyle\endcsname\relax
  \providecommand{\doi}[1]{doi: #1}\else
  \providecommand{\doi}{doi: \begingroup \urlstyle{rm}\Url}\fi

\bibitem[Chen(2012)]{chen:match12}
N.~Chen.
\newblock On computing {P}areto stable assignments.
\newblock In \emph{Proceedings of the 29th International Symposium on
  Theoretical Aspects of Computer Science}, pages 384--395, March 2012.

\bibitem[Chen and Ghosh(2010)]{chen+g:match}
N.~Chen and A.~Ghosh.
\newblock Algorithms for {P}areto stable assignment.
\newblock In \emph{Proceedings of the Third International Workshop on
  Computational Social Choice}, pages 343--354, September 2010.

\bibitem[Domani\c{c} et~al.(2016)Domani\c{c}, Lam, and
  Plaxton]{domanic+lp:match}
N.~O. Domani\c{c}, C.-K. Lam, and C.~G. Plaxton.
\newblock Strategyproof {P}areto-stable mechanisms for two-sided matching with
  indifferences.
\newblock Technical Report TR--16--19, Department of Computer Science,
  University of Texas at Austin, November 2016.

\bibitem[Erdil and Ergin()]{erdil+e:match}
A.~Erdil and H.~Ergin.
\newblock Two-sided matching with indifferences.
\newblock Working paper.

\bibitem[Erdil and Ergin(2008)]{erdil+e:tie}
A.~Erdil and H.~Ergin.
\newblock What's the matter with tie-breaking? {I}mproving efficiency in school
  choice.
\newblock \emph{American Economic Review}, 98:\penalty0 669--689, 2008.

\bibitem[Fredman and Tarjan(1987)]{FT87}
M.~L. Fredman and R.~E. Tarjan.
\newblock Fibonacci heaps and their uses in improved network optimization
  algorithms.
\newblock \emph{J. ACM}, 34:\penalty0 596--615, 1987.

\bibitem[Gale and Shapley(1962)]{gale+s:match}
D.~Gale and L.~S. Shapley.
\newblock College admissions and the stability of marriage.
\newblock \emph{American Mathematical Monthly}, 69:\penalty0 9--15, 1962.

\bibitem[Kamiyama(2014)]{kamiyama:match14}
N.~Kamiyama.
\newblock A new approach to the {P}areto stable matching problem.
\newblock \emph{Mathematics of Operations Research}, 39:\penalty0 851--862,
  2014.

\bibitem[Kesten(2010)]{Kes10}
O.~Kesten.
\newblock School choice with consent.
\newblock \emph{The Quarterly Journal of Economics}, 125\penalty0 (3):\penalty0
  1297--1348, 2010.

\bibitem[Kuhn(1955)]{Kuhn55}
H.~W. Kuhn.
\newblock The {H}ungarian method for the assignment problem.
\newblock \emph{Naval Research Logistics Quarterly}, 2:\penalty0 83--97, 1955.

\bibitem[Manlove(2013)]{manlove:match}
D.~F. Manlove.
\newblock \emph{Algorithmics of Matching Under Preferences}.
\newblock World Scientific, Singapore, 2013.

\bibitem[Roth(1982)]{roth:match82}
A.~E. Roth.
\newblock The economics of matching: Stability and incentives.
\newblock \emph{Mathematics of Operations Research}, 7:\penalty0 617--628,
  1982.

\bibitem[Roth(1985)]{roth:match85}
A.~E. Roth.
\newblock The college admissions problem is not equivalent to the marriage
  problem.
\newblock \emph{Journal of Economic Theory}, 36:\penalty0 277--288, 1985.

\bibitem[Roth and Sotomayor(1990)]{roth+s:match}
A.~E. Roth and M.~Sotomayor.
\newblock \emph{Two-Sided Matching: A Study in Game-Theoretic Modeling and
  Analysis}.
\newblock Cambridge University Press, New York, 1990.

\bibitem[Shapley and Shubik(1972)]{shapley+s:assign}
L.~S. Shapley and M.~Shubik.
\newblock The assignment game {I}: The core.
\newblock \emph{International Journal of Game Theory}, 1:\penalty0 111--130,
  1972.

\bibitem[Sotomayor(2011)]{sotomayor:pareto-stable}
M.~Sotomayor.
\newblock The {P}areto-stability concept is a natural solution concept for
  discrete matching markets with indifferences.
\newblock \emph{International Journal of Game Theory}, 40:\penalty0 631--644,
  2011.

\end{thebibliography}

\appendix

\section{Discussion of the Two-Phase Approach} \label{app:phase}

As discussed in Section~\ref{sec:intro}, Erdil and
Ergin~\cite{erdil+e:match} present a polynomial-time algorithm for
computing a Pareto-stable matching for a given instance of SMCW (and
also its generalizations to SMIW and CAW).  Their algorithm uses a
two-phase approach: in the first phase, ties are broken arbitrarily,
and the Gale-Shapley DA algorithm is used to obtain a weakly stable
matching; in the second phase, the matching is repeatedly updated via
a sequence of Pareto improvements until no such improvement is
possible.  This two-phase framework was previously proposed by
Sotomayor~\cite{sotomayor:pareto-stable}, who argued its correctness by
observing that when we apply a Pareto improvement to a weakly stable
matching, we obtain another weakly stable matching.

It is natural to ask whether there is a strategyproof Pareto-stable
mechanism for SMCW based on the foregoing two-phase approach.  More
precisely, suppose the men and women are indexed from $1$ to $n$, and
assume that we break ties in the first phase in favor of
higher-indexed agents.  Is there a way to implement the second phase
so that the resulting two-phase algorithm corresponds to a
strategyproof Pareto-stable mechanism?  The example presented below
provides a negative answer to this question.

Consider an SMCW instance $I$ with men $\{p_1,p_2,p_3\}$ and women
$\{q_1,q_2,q_3\}$, and where the preferences of the agents are as
follows: $p_1$ prefers $q_2$, then $q_3$, then $q_1$; $p_2$ prefers
$q_1$, then $q_3$, then $q_2$; $p_3$ is indifferent between $q_1$ and
$q_2$, and prefers $q_1$ and $q_2$ to $q_3$; $q_1$ prefers $p_3$, then
$p_1$, then $p_2$; $q_2$ is indifferent between all of the men; $q_3$
prefers $p_3$, then $p_2$, then $p_1$.  Let $M_1$ through $M_6$ denote
the six possible matchings: $M_1=\{(p_1,q_1),(p_2,q_2),(p_3,q_3)\}$;
$M_2=\{(p_1,q_1),(p_2,q_3),(p_3,q_2)\}$;
$M_3=\{(p_1,q_2),(p_2,q_1),(p_3,q_3)\}$;
$M_4=\{(p_1,q_2),(p_2,q_3),(p_3,q_1)\}$;
$M_5=\{(p_1,q_3),(p_2,q_1),\Formatting{\allowbreak} (p_3,q_2)\}$;
$M_6=\{(p_1,q_3),(p_2,q_2),(p_3,q_1)\}$.  It is easy to verify that
$\{M_2,M_4,M_5\}$ is the set of weakly stable matchings for $I$.
(Matchings $M_1$ and $M_3$ are blocked by $(p_3,q_1)$, and matching
$M_6$ is blocked by $(p_2,q_3)$.)  Furthermore, the set of
Pareto-stable matchings for $I$ is $\{M_4,M_5\}$.  (Matching $M_2$ is
Pareto-dominated by matching $M_4$.)  If we break ties in favor of the
agents with higher indices, then it is easy to verify that the first
phase produces matching $M_5$.  Since $M_5$ is Pareto-stable, the
second phase does not update the matching, and hence $M_5$ is the
final output.

Now suppose man $p_1$ lies by stating that he prefers $q_2$, then
$q_1$, then $q_3$, and let $I'$ denote the resulting SMCW instance.
It is easy to verify that $\{M_2,M_4\}$ is the set of weakly stable
matchings for $I'$.  (Matchings $M_1$ and $M_3$ are blocked by
$(p_3,q_1)$, matching $M_5$ is blocked by $(p_1,q_1)$, and matching
$M_6$ is blocked by $(p_2,q_3)$.)  Furthermore, the set of
Pareto-stable matchings for $I'$ is $\{M_4\}$.  (Matching $M_2$ is
Pareto-dominated by matching $M_4$.)  Thus $M_4$ is the only possible
output of the second phase.  Since man $p_1$ prefers his match under
$M_4$ to his match under $M_5$, strategyproofness is violated.

% \section{Deleted Stuff}

% \input{deleted}

%  LocalWords:  Indifferences MWMCM UAP MCM MWMCMs nonnegative UAPs
%  LocalWords:  multibidder
%  LocalWords:  Multibidders
%  LocalWords:  IUAP IUAPs
%  LocalWords:  subgraph vertices

%%% Local Variables:
%%% mode: latex
%%% TeX-master: "tr"
%%% End:

\end{document}